\documentclass[runningheads]{llncs}
\usepackage[T1]{fontenc}
\usepackage{xcolor}
\usepackage{sidenotes}

\usepackage{graphicx}
\usepackage{amssymb}
\usepackage{amsmath}
\usepackage{xcolor}
\usepackage{longtable}
\usepackage{caption}
\usepackage{algorithm}
\usepackage{algpseudocode}
\usepackage{longtable}
\usepackage{nicefrac}
\usepackage{float}

\algrenewcommand{\algorithmiccomment}[1]{%
	\hfill\textcolor{gray}{\footnotesize $\triangleright$ #1}%
}

\newcommand{\Proc}[2]{\textproc{#1}(#2)}

\usepackage{listings}

\lstdefinelanguage{C}{
	keywords={double, False,True,None,and,as,assert,break,class,continue,def,del,elif,
		else,except,finally,for,from,global,if,import,in,is,lambda,nonlocal,not,
		or,pass,raise,return,try,while,with,yield},
	keywordstyle=\color{blue}\bfseries,
	ndkeywords={self},
	ndkeywordstyle=\color{purple},
	identifierstyle=\color{black},
	commentstyle=\color{gray}\itshape,
	stringstyle=\color{orange},
	sensitive=true
}

\usepackage{bbding}
\usepackage{subfiles}
\usepackage{stmaryrd}
\usepackage{bbm}
\usepackage{xspace}
\usepackage{cleveref}

 \usepackage{pgfplots}
 \pgfplotsset{compat=1.18}
 \usepackage{tikz}
 \usetikzlibrary{arrows.meta}
 
 \definecolor{colorA}{HTML}{e3e6ff} 
 \definecolor{colorC}{HTML}{eaffea} 
 \definecolor{colorE}{HTML}{e0ffff} 
 \definecolor{colorBoundary}{HTML}{a0a0a0} 
 \definecolor{colorMark}{HTML}{bf212f}

\crefname{algorithm}{Algorithm}{Algorithms}
\Crefname{algorithm}{Algorithm}{Algorithms}
\crefname{proposition}{Prop.}{Props.}
\Crefname{proposition}{Prop.}{Props.}
\crefname{section}{\S}{\S\S}
\crefname{table}{Table}{Tables}

\DeclareMathOperator{\AST}{\mathrm{AST}}
\DeclareMathOperator{\IF}{\mathrm{If}}
\DeclareMathOperator{\Unif}{\mathrm{U}}
\DeclareMathOperator{\Tri}{\mathrm{T}}

\newcommand{\ie}{\textit{i.e}.\xspace}
\newcommand{\viz}{\textit{viz}.\xspace}
\newcommand{\eg}{\textit{e.g}.\xspace}
\newcommand{\resp}{resp.\xspace}
\newcommand{\wrt}{w.r.t.\xspace}

\newcommand{\Var}{\mathrm{Var}}
\newcommand{\R}{\mathbb{R}}

\newcommand{\cp}{\mathsf{copy}}
\newcommand{\sw}{\mathsf{swap}}
\newcommand{\sem}[1]{\llbracket #1 \rrbracket}
\newcommand{\dsem}[1]{\llparenthesis #1 \rrparenthesis}
\newcommand{\Prob}{\mathcal{P}}
\newcommand{\Pro}{\mathbb{P}}
\newcommand{\inv}{^{-1}}
\newcommand{\one}{\mathbbm{1}}
\newcommand{\ite}[3]{\IF[#1,#2,#3]}
\newcommand{\unif}[2]{\Unif(#1,#2)}
\newcommand{\tri}[2]{\Tri(#1,#2)}

\begin{document}
	\title{Exact Evaluation of Probabilistic Programs with Cylindrical Algebraic Decomposition}
	\titlerunning{Exact Evaluation of Probabilistic Programs with CAD}
	%
	\author{Mohamed Hamza Bandukara \and Fredrik Dahlqvist \and Niki Omidvari\textsuperscript{(\Envelope)} 
	}
	\authorrunning{Bandukara, Dahlqvist and Omidvari}
	%
	\institute{Queen Mary University of London, London, UK\\
		\email{\{m.h.bandukara,f.dahlqvist,n.omidvari\}@qmul.ac.uk}}
	\maketitle              
	\begin{abstract}
		We present a method for computing the exact output distribution of small programs with random inputs. Specifically, we are interested in inline programs manipulating sensor data such as \eg GPS or inertial measurement sensors whose inputs have a known or well-modelled distribution. These programs typically only include relatively few variables, arithmetic operations, square roots and if-else statements. This small syntax allows us to recast the problem of computing the exact output distribution as a cylindrical algebraic decomposition problem followed by symbolic and/or numerical integration. We present this method in detail and show with two prototypes that it can successfully be applied to benchmarks from the literature on floating-point arithmetic and small programs from open-source sensor libraries.
		
	\end{abstract}

	\section{Introduction}\label{sec:intro}

There are two ways in which randomness can find its way into a computer program. First, it can be inserted there by design, through sampling from a Pseudo-Random Number Generator (PRNG). This kind of randomness is used in a huge range of applications, most notably in Monte-Carlo simulations for tasks as varied as pricing financial instruments, weather forecasting or computer graphics. Second, it can enter the program via interactions with the outside world, for example through sensors forming part of a cyber-physical system like an autonomous vehicle, an industrial robot or a mobile phone. All sensors are necessarily noisy, and any program processing sensor data will therefore exhibit some randomness. Importantly, this noise can be -- and often is -- modelled to a good level of precision, \eg~\cite{cutler2004dutton,hofmann2003navigation,Specht_2021} for GPS sensors or \cite{electronics12112458,wu2017dynamic} for inertial measurement sensors (gyroscopes, accelerometers, magnetometers).

Irrespective of the source of randomness, we will refer to a program exhibiting randomness as a \emph{probabilistic program}. Our objective in this work is to compute the exact output distribution of probabilistic programs with known continuous input distributions\footnote{We leave the addition of discrete and mixed distributions to future work.}.  In principle, our method can be applied to any probabilistic program. In practice, it is aimed at programs manipulating sensor data. Computing the output distribution of a probabilistic program is a very hard problem in general and we therefore restrict our attention to small programs containing ($i$) arithmetic operations, ($ii$) radicals, and ($iii$) if-else statements. This syntax is sufficient to capture many functions in the code bases of various sensor libraries, for example associated with Inertial Measurement Units (IMUs) or GPS sensors.  Below is a fairly typical example from an IMU library\footnote{\url{https://github.com/CCNYRoboticsLab/imu\_tools}}; based on an input acceleration from an accelerometer it returns the first component of a quaternion representing the orientation of the global w.r.t.\ the local frame.
\begin{lstlisting}
	double getMeasurement(double ax, double ay, double az){
		normalizeVector(ax, ay, az);
		if (az >= 0){return sqrt((az + 1) * 0.5);} 
		else {return -ay / (2.0 * sqrt((1 - az) * 0.5);}}
\end{lstlisting}
We focus on relatively small probabilistic programs like this one (included as \texttt{get\_meas} in our results, see~\cref{tab:benchmarks}) rather than the typically much more complex code bases of Monte-Carlo simulations.

There is of course a very simple way of characterising the output distribution of a probabilistic program: sample repeatedly from the input distributions, run the program on these inputs, collect the outputs, and plot a histogram. In other words, analyse the program by Monte-Carlo simulation. This is a frequent and conceptually simple approach to testing probabilistic programs. This simplicity however comes at the expense of \emph{precision} and \emph{formal guarantees}.

\noindent\emph{Precision}. By definition, Monte-Carlo simulation can only \emph{approximate} the probability of a range of output values. By the central limit theorem, the error incurred by this approximation only decreases as $1/\sqrt{n}$, where $n$ is the number of samples. Thus what may at first seem like a computationally cheap method can turn out to be quite expensive if a high level of precision is required, particularly if the probability of a rare event must be approximated. Beyond the computational cost, the problem of generating a lot of quality random inputs isn't trivial. There are two options: either sample analog data directly from the sensor under consideration, the throughput of samples is then typically limited, or use a much faster PRNG, in which case it is typically impossible to exclude the presence of statistical artifacts beyond a small number of samples (\eg the original Mersenne twister~\cite{matsumoto1998mersenne} only guarantees the i.i.d. behaviour of successive samples up to $n=623$). This phenomenon is illustrated perfectly by the poor performance of `good' PRNGs on physics problem requiring many samples~\cite{bad_randomness}.

\noindent\emph{Formal guarantees}. Monte-Carlo simulations cannot offer formal guarantees of the type `\emph{the probability of the output lying in a set $X$ of bad outcomes is less that $\alpha\%$}'. This is partly due to the impossibility of ruling out statistical artifacts when using PRNGs, but even when samples are \textit{bona fide} i.i.d. samples, it is extremely difficult to provide useful upper-bounds to the distance between the empirical output distribution (based on samples) and the actual output distribution. 

Computing exact output distributions is much harder than Monte-Carlo simulation but it solves the precision problem and can also be used to provide formal guarantees in the development and certification of cyber-physical systems. As we shall show in \cref{sec:experiments}, computing exact/analytical output distributions is a practical proposal for benchmarks from the floating-point analysis community (FPBench~\cite{fpbench}) and small programs from open-source IMU sensors libraries. 

The structure of this paper is as follows: in \cref{sec:prelim} we will define arithmetic on random variables, Cylindrical Algebraic Decomposition (CAD), and show how the latter can be used to compute the former. In \cref{sec:programs} we will describe the syntax and semantics of the programs that we consider in this work. This probabilistic semantics (the output distribution)  is concretely evaluated using algorithms described in \cref{sec:algorithms}. Finally, we present an empirical evaluation in \cref{sec:experiments}.

\subsection*{Related work.} 

\paragraph{Statistical models.} Arithmetic operations and if-else branching on random variables is used by the statistical community to develop generative models, \eg 
\begin{align*}
	X&\sim\mathcal{N}(0,1) &
	Y&=\begin{cases}
		\alpha_1\sqrt{X} + \epsilon & X\geq 0 \\
		\alpha_2 X + \epsilon & X<0
	\end{cases}\qquad\quad \epsilon \sim \mathcal{N}(0,\sigma).
\end{align*}
These models will be implemented as probabilistic programs.
Methods for evaluating the Probability Density Function (PDF) of such models are developed in \cite{pacal2012,korzen2014pacal}. However, their method is purely numerical and cannot deal with arithmetic expressions containing repeated variables. Although we do not focus on this application, our method can easily be applied to statistical models.

\paragraph{Probabilistic floating-point arithmetic.}
The literature on floating-point arithmetic for random inputs \cite{bouissou2012generalization,bouissou2016uncertainty,lohar2019sound,sankaranarayanan2020reasoning,constantinides2021rigorous,constantinides2025automated} is closely related to our work, it targets the same kind of programs under the same condition of random inputs. The key differences with our work are that (a) this line of work does not evaluate output distributions exactly, relying on under- and over-approximations like $p$-boxes~ \cite{bouissou2012generalization,constantinides2021rigorous}, and (b) we do not (yet!) consider rounding errors when computing exact output distributions and work instead in infinite precision (symbolically). 

\paragraph{Exact inference in probabilistic programs.}
Though our focus is on small in-line probabilistic programs without inference, there is broader work on exact inference for probabilistic programs that includes in particular computing output distributions (the \emph{evidence} in the Bayesian perspective).
For discrete distributions, exact inference is more tractable and can be done using weighted model counting, as in Dice \cite{holtzen2020scaling}. Continuous distributions can be discretised to exploit these methods; the recent `bit-blasting' approach \cite{garg2024bit} being a particularly efficient example of this approach. Being approximate, these discretisation methods cannot offer the formal guarantees we aim to compute. 

Computer algebra can be used to simplify the distributional representations of probabilistic programs, though they may leave results partially unevaluated or finally use numeric methods. 
Hakaru \cite{narayanan2016probabilistic,carette2016simplifying} uses the Maple computer algebra system for simplification but may still rely on sampling for full evaluation. Similarly, \cite{bhat2013deriving} uses a type-based system to represent probability distributions but uses numerical integration to evaluate them, and so does not produce closed-form solutions.
The top solver for exact symbolic inference is PSI \cite{gehr2016psi,gehr2020lambdapsi}, which supports both discrete and continuous distributions and a rich syntax including conditioning, bounded loops, and higher-order functions. PSI works by incrementally building, simplifying, and integrating the joint probability distribution of a program.
PSI is a powerful general-purpose tool targeting a very expressive language, and may therefore return results containing unevaluated integrals or overly complex PDFs. We have designed two tools, one to compute closed-form Cumulative Distribution Functions (CDFs) and PDFs, based on Mathematica's symbolic integration engine, the other, based on C\texttt{++} libraries, to evaluate output CDFs and PDFs without necessarily constructing a full closed-form solution.

\paragraph{Probabilistic Circuits}~\cite{choi2020probabilistic}, such as sum-product networks \cite{poon2011sum,sanchez2021sum} and arithmetic circuits \cite{darwiche2003differential}, represent joint distributions as trees of products and weighted sums of density functions. Their structural constraints enable exact polynomial-time inference for queries such as marginalisation and conditioning but can lead to large circuits with many repetitions. 
They are designed for queries on joint distributions, rather than for computing output distributions. The SPPL system \cite{saad2021sppl} bridges this gap by translating probabilistic programs into sum-product expressions, enabling exact symbolic inference comparable to PSI. SPPL is faster than PSI on shared benchmarks but targets a much less expressive language (\eg it cannot compute $Z=X\div Y$ or accept infinite-support distributions).

	\section{Preliminaries}\label{sec:prelim}

In this section we describe the problem of evaluating arithmetic expressions over random variables, and show how cylindrical algebraic decomposition (CAD) can solve it. We apply this technique in the next sections to evaluate the output distributions, \ie the \emph{probabilistic semantics}, of probabilistic programs.

\paragraph{Arithmetic on random variables.}
We assume familiarity with the notions of random variable, probability measure, and Probability Density (\resp  Cumulative Distribution) Functions, \ie PDF (\resp CDF). By \emph{arithmetic expression} we mean an expression given by the grammar:
\begin{align*}
	e & ::= X \in \text{Var} \mid c \in \mathbb{R} \mid e + e \mid e - e \mid e \times e \mid e \div e
\end{align*}
where $\Var=\{X_1, \ldots, X_n\}$ is a finite set of variable symbols. A \emph{polynomial expression}, or simply a \emph{polynomial}, is an arithmetic expression not containing the division symbol $\div$.

\emph{Independent variables.} Consider first the case of two \emph{independent} real random variables $X_1,X_2$ described by PDFs $f_1$ and $f_2$ respectively. It is well known~\cite{springer1979} that the density functions of the distributions of $X_1+X_2, X_1-X_2, X_1\times X_2$ and $X_1\div X_2$ can be expressed analytically as follows: 
{\small
	\begin{align}
		f_{X_1+X_2}(t) &= \int_{-\infty}^\infty{f_1(x)f_2(t-x) \, dx} \quad&
		f_{X_1-X_2}(t) &= \int_{-\infty}^\infty{f_1(x)f_2(x-t) \, dx} \nonumber \\
		f_{X_1 \times X_2}(t) &= \int_{-\infty}^\infty{\frac{1}{\left | x \right | }f_1(x)f_2\hspace{-3pt}\left(\frac{t}{x}\right) \, dx} \quad &
		f_{X_1 \div X_2}(t) &= \int_{-\infty}^\infty{\left | y \right | f_1(ty)f_2(y) \, dy} \label{eq:independent}
\end{align}}The independence assumption is crucial as it  means that the joint density of $X_1,X_2$ factors as the product of the densities, \viz $f_{X_1,X_2}(x,y)=f_1(x)f_2(y)$, which makes the derivations of the expressions \eqref{eq:independent} possible. Evaluating independent arithmetic expressions over random expression numerically is discussed in~\cite{pacal2012} and implemented in the Python library PaCal~\cite{korzen2014pacal}.

 The density function of any arithmetic expression over independent random variables $X_1, \ldots, X_n$ can thus be computed inductively using  \eqref{eq:independent} \emph{provided that no variable is repeated}. Indeed, variable repetition immediately breaks independence: consider the expression $X\div (X+Y)$, the random variables $X$ and $X+Y$ are not independent since what we usually mean when writing $X\div (X+Y)$ is that the two occurrences of $X$ represent the \emph{same} random element. Put differently, to sample from $X\div (X+Y)$ we sample $x$ from $X$, $y$ from $Y$, \emph{copy $x$} and evaluate $x\div(x+y)$ (using our two copies of $x$).

\emph{Non-independent variables.} If we cannot assume that variables are independent, we must assume a general joint density $f_{X_1,X_2}(x_1,x_2)$ that cannot be factored as above. To evaluate the distribution of $X_1\ast X_2$ for $\ast\in\{+,-,\times,\div\}$ one can compute the CDF as $	F_{X_1\ast X_2}(t)=\int_{\{(x,y)\mid x\ast y\leq t\}} f_{X_1,X_2}(x_1,x_2)~dx_1dx_2$
and, if it is analytically tractable,  take its derivative to get the PDF. More generally, for an arithmetic expression $\varphi(X_1,\ldots,X_n)$ over non-independent random variables with joint density $f(x_1, \ldots, x_n)$, the CDF of $\varphi(X_1,\ldots,X_n)$ is given by
\begin{align}
	F(t)=\int_{\{(x_1, \ldots, x_n)\mid \varphi(x_1,\ldots, x_n)\leq t\}} f(x_1,\ldots,x_n)~dx_1\ldots dx_n. \label{eq:general_integral}
\end{align}
To concretely evaluate integrals of this shape we will exploit the fact that the set over which the integral above is taken has a very specific shape known as a \emph{semi-algebraic set}. Although it is not essential for our approach to work, we will also assume throughout that the $n$ input variables are independent, meaning that $f(x_1, \ldots, x_n)=f_{1}(x_1)\ldots f_{n}(x_n)$. However, this does not mean that we can fall back on \eqref{eq:independent}, as $\varphi(X_1, \ldots,X_n)$ will typically contain variable repetitions. 

\paragraph{Cylindrical Algebraic Decomposition (CAD).}
The class of \textit{semi-algebraic sets}~\cite[\S 2]{basu2006algorithms} of $\R^n$ is the smallest class of subsets $S\subseteq \R^n$ that contains ($i$) \emph{algebraic sets}, \ie sets $\{x \in\R^n\mid p(x)=0\}$ of zeroes of a polynomial $p$, ($ii$) sets $\{x\in\R^n\mid p(x)>0\}$ defined by polynomial inequalities, and ($iii$) is closed under intersection, union and complementation.  A function $f: S\to T$ between two semi-algebraic sets $S\subseteq\R^{m}$ and $T\subseteq\R^n$ is \emph{a semi-algebraic function} if its graph $\{(x,f(x))\mid x\in \R^m\}$ is semi-algebraic in $\R^{m+n}$. It is easy to check that the sets over which integrals of the shape~\eqref{eq:general_integral} are taken are semi-algebraic.

\begin{proposition}\label{prop:division_free}
	For any arithmetic expression $e(X_1,\ldots, X_k)$, the set $\{x\in \R^k\mid e(x)<0\}$ is semi-algebraic.
\end{proposition}
\begin{proof}
	If the abstract syntax tree $\AST(e)$ of $e$ does not contain a division node, then there is nothing to show, so let us assume that it does.
	By replacing all subterms $e_1 - e_2$ by $e_1 + (-1\cdot e_2)$, we can assume w.l.o.g. that $\AST(e)$ does not contain any subtraction. Moreover, by replacing all subterms $e_1 \times (e_2\div e_3)$ by $(e_1\times e_2)\div e_3$, and subterms $e_1+(e_2\div e_3)$ by $((e_1\times e_3)+e_2)\div e_3$ (and similarly for addition on the right) we rewrite $e$ in such a way that no division occurs below an addition or a multiplication in $\AST(e)$. Next, we get rid of consecutive divisions in $\AST(e)$ by re-writing subterms $e_1\div (e_2\div e_3)=(e_1\times e_3)\div e_2$ and $(e_1\div e_2)\div e_3$ as $e_1\div (e_2\times e_3)$. If need be,  we push any division appearing under one of these multiplications up in the manner described above. It follows that we can rewrite $e$ in such a way that the only division in $\AST(e)$ is the top node, \ie that $e=e_1\div e_2$. We now simply write
	\begin{align*}
		\{x\in \R^k\mid e_1(x)\div e_2(x)<0\}=&\{x\in \R^k\mid e_2(x)>0 \wedge e_1(x)<0\}~\cup \\
		&\{x\in \R^k\mid e_2(x)<0 \wedge e_1(x)>0\}
	\end{align*}
	which is clearly semi-algebraic.\hfil\qed
\end{proof}

We can now define the central concept of \emph{Cylindrical Algebraic Decomposition} or CAD, see \eg~\cite{arnon1984cylindrical},\cite[\S 5]{basu2006algorithms}. A CAD of $\R^n$ is a sequence $\mathcal{S}_1, \ldots, \mathcal{S}_n$ of finite partitions $\mathcal{S}_i$ of $\R^i, 1\leq i\leq n$ into semi-algebraic subsets called \emph{cells of level $i$} or \emph{$i$-cells}. These partitions must satisfy the following properties:
\begin{enumerate}
	\item Each cell of level 1 in $\mathcal{S}_1$ is either a point or an open interval
	\item For each $i$-cell $S\in\mathcal{S}_i, 1\leq i<n$, there are finitely many continuous semi-algebraic functions $\xi_{S,1}< \ldots< \xi_{S,k_S}: S\to \R$ such that the \emph{cylinder} $S\times \R\subseteq \R^{i+1}$ is the disjoint union of the $2k_S+1$ cells of level $i+1$ given by:
	\item[-] $\{(x_1, \ldots,x_i, x_{i+1})\mid (x_1, \ldots,x_i)\in S, x_{i+1}=\xi_{S,j}(x_1,  \ldots,x_i)\}$, that is to say the graph of the functions $\xi_{S,j}, 1\leq j\leq k_S$, or
	\item[-] $\{(x_1, \ldots,x_i, x_{i+1})\mid (x_1, \ldots,x_i)\in S, \xi_{S,j}(x_1, \ldots,x_i)<x_{i+1}<\xi_{S,j+1}(x_1, \ldots,x_i)\}$, that is to say
	the band bounded from above and below by the graphs of $\xi_{S,j}$ and $\xi_{S,j+1}$ for $0\leq j\leq k_S$ where $\xi_{S,0}=-\infty$ and $\xi_{S,k_S+1}=\infty$.
\end{enumerate}
Let $\pi_i: \R^{i+1}\simeq \R^i\times \R\to \R^i$ be the canonical projection forgetting the last element in an $i+1$-tuple.
Given a cell $S$ of level $i$, the cells of level $i+1$ decomposing the cylinder $S\times \R=\pi_i\inv(S)$ are called the cells \emph{above} $S$ since they project down into $S$. Note in particular that if $k_S=0$ in the definition above, then the only cell above $S$ is the entire cylinder $S\times \R$.

A CAD $(\mathcal{S}_i)_{1\leq i\leq n}$ is said to be \emph{adapted to a finite family $T_1, \ldots, T_k$ of semi-algebraic sets} if each $T_j, 1\leq j\leq k$ is a union of cells of $(\mathcal{S}_i)_{1\leq i\leq n}$. A CAD $(\mathcal{S}_i)_{1\leq i\leq n}$ is said to be \emph{adapted to a finite family $P_1, \ldots, P_k$ of polynomials in the variables $X_1, \ldots, X_n$} if for each $n$-cell $S$ and each polynomial $P_i$, the sign of $P_i$ is constant on $S$ (\ie >0, <0 or =0). 

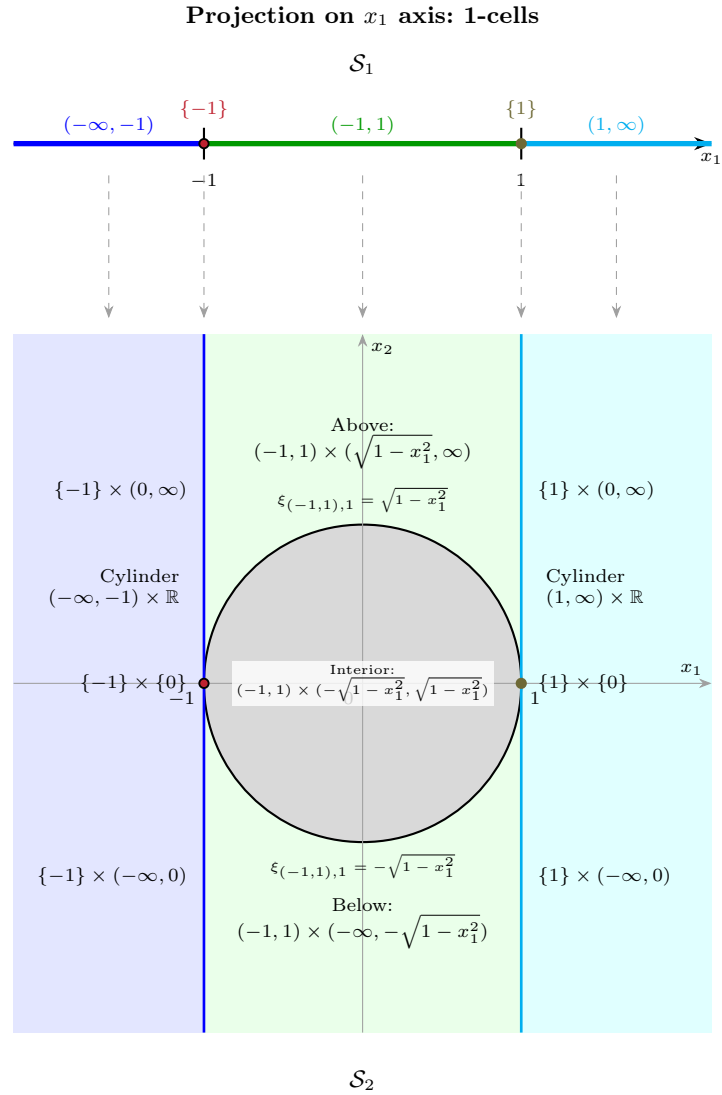
\begin{figure}[htbp]
	\centering
	\begin{tikzpicture}[
		font=\footnotesize, 
		dot/.style = {circle, draw, fill=colorMark, inner sep=0pt, minimum size=3.5pt, line width=0.7pt},
		>=Stealth
		]
		
		\newcommand{\myunitlength}{2.1cm} 
		
		\begin{axis}[
			x=\myunitlength,            
			y=\myunitlength,            
			xmin=-2.2, xmax=2.2,      
			ymin=-2.2, ymax=4.2, 
			axis lines=none,      
			clip=false
			]
			
			
			\fill[colorA] (axis cs:-2.2,-2.2) rectangle (axis cs:-1,2.2);
			\fill[colorC] (axis cs:-1,-2.2) rectangle (axis cs:1,2.2);
			\fill[colorE] (axis cs:1,-2.2) rectangle (axis cs:2.2,2.2);
			
			\fill[white] (axis cs:0,0) circle [radius=1];
			\fill[opacity=0.15, black] (axis cs:0,0) circle [radius=1];
			\draw[black, thick] (axis cs:0,0) circle [radius=1];
			
			\addplot[blue, line width=1pt] coordinates {(-1,-2.2) (-1,2.2)};
			\addplot[cyan, line width=1pt] coordinates {(1,-2.2) (1,2.2)};
			
			\draw[->, colorBoundary] (axis cs:-2.2,0) -- (axis cs:2.2,0) node[above left, text=black, font=\scriptsize] {$x_1$};
			\draw[->, colorBoundary] (axis cs:0,-2.2) -- (axis cs:0,2.2) node[below right, text=black, font=\scriptsize] {$x_2$};
			
			\node[below left, font=\scriptsize] at (axis cs:0,0) {$0$};
			\node[below left, font=\scriptsize] at (axis cs:-1,0) {$-1$};
			\node[below right, font=\scriptsize] at (axis cs:1,0) {$1$};
			
			\node[font=\small] at (axis cs:0, -2.5) {$\mathcal{S}_2$};
			
			\node[dot] (pB) at (axis cs:-1,0) {};
			\node[dot, color=olive!60!black, fill=olive!60!black] (pD) at (axis cs:1,0) {};
			
			\node[anchor=east, align=right, font=\scriptsize] at (axis cs:-1.1, 0.6) {Cylinder\\$(-\infty,-1) \times \mathbb{R}$};
			\node[anchor=west, align=left, font=\scriptsize] at (axis cs:1.1, 0.6) {Cylinder\\$(1,\infty) \times \mathbb{R}$};
			
			\node[anchor=east, font=\scriptsize, align=right] at (axis cs:-1.05, 0) {$\{-1\}\times\{0\}$};
			\node[anchor=south east, font=\scriptsize, align=right] at (axis cs:-1.05, 1.1) {$\{-1\}\times(0,\infty)$};
			\node[anchor=north east, font=\scriptsize, align=right] at (axis cs:-1.05, -1.1) {$\{-1\}\times(-\infty,0)$};
			
			\node[anchor=west, font=\scriptsize, align=left] at (axis cs:1.05, 0) {$\{1\}\times\{0\}$};
			\node[anchor=south west, font=\scriptsize, align=left] at (axis cs:1.05, 1.1) {$\{1\}\times(0,\infty)$};
			\node[anchor=north west, font=\scriptsize, align=left] at (axis cs:1.05, -1.1) {$\{1\}\times(-\infty,0)$};
			
			\node[anchor=south, font=\scriptsize, align=center] at (axis cs:0, 1.3) {Above:\\$(-1,1)\times(\sqrt{1-x_1^2},\infty)$};
			\node[anchor=north, font=\scriptsize, align=center] at (axis cs:0, -1.3) {Below:\\$(-1,1)\times(-\infty,-\sqrt{1-x_1^2})$};
			
			\node[font=\tiny, fill=white, fill opacity=0.8, text opacity=1, inner sep=1.5pt, align=center] at (axis cs:0, 0) {Interior:\\$(-1,1)\times(-\sqrt{1-x_1^2},\sqrt{1-x_1^2})$};
			\node[anchor=south, font=\tiny] at (axis cs:0, 1.02) {$\xi_{(-1,1),1}=\sqrt{1-x_1^2}$};
			\node[anchor=north, font=\tiny] at (axis cs:0, -1.02) {$\xi_{(-1,1),1}=-\sqrt{1-x_1^2}$};

			
			\node[font=\small] at (axis cs:0, 4.2) {\textbf{Projection on $x_1$ axis: 1-cells}};
			\node[font=\small] at (axis cs:0, 3.9) {$\mathcal{S}_1$};
			
			\draw[->, thick] (axis cs:-2.2, 3.4) -- (axis cs:2.2, 3.4) node[below, font=\scriptsize] {$x_1$};
			
			\draw[blue, line width=2pt] (axis cs:-2.2, 3.4) -- (axis cs:-1, 3.4);
			\draw[color=green!60!black, line width=2pt] (axis cs:-1, 3.4) -- (axis cs:1, 3.4);
			\draw[cyan, line width=2pt] (axis cs:1, 3.4) -- (axis cs:2.2, 3.4);
			
			\draw[thick] (axis cs:-1, 3.5) -- (axis cs:-1, 3.3) node[below=2pt, font=\scriptsize] {$-1$};
			\draw[thick] (axis cs:1, 3.5) -- (axis cs:1, 3.3) node[below=2pt, font=\scriptsize] {$1$};
			
			\node[dot] at (axis cs:-1, 3.4) {};
			\node[dot, color=olive!60!black, fill=olive!60!black] at (axis cs:1, 3.4) {};
			
			\node[above, blue, font=\scriptsize] at (axis cs:-1.6, 3.4) {$(-\infty,-1)$};
			\node[above, color=colorMark, font=\scriptsize] at (axis cs:-1.0, 3.5) {$\{-1\}$};
			\node[above, color=green!60!black, font=\scriptsize] at (axis cs:0, 3.4) {$(-1,1)$};
			\node[above, color=olive!60!black, font=\scriptsize] at (axis cs:1.0, 3.5) {$\{1\}$};
			\node[above, cyan, font=\scriptsize] at (axis cs:1.6, 3.4) {$(1,\infty)$};

			\draw[->, dashed, colorBoundary] (axis cs:-1.6, 3.2) -- (axis cs:-1.6, 2.3);
			\draw[->, dashed, colorBoundary] (axis cs:-1.0, 3.2) -- (axis cs:-1.0, 2.3);
			\draw[->, dashed, colorBoundary] (axis cs:0.0, 3.2)  -- (axis cs:0.0, 2.3);
			\draw[->, dashed, colorBoundary] (axis cs:1.0, 3.2)  -- (axis cs:1.0, 2.3);
			\draw[->, dashed, colorBoundary] (axis cs:1.6, 3.2)  -- (axis cs:1.6, 2.3);
			
		\end{axis}
	\end{tikzpicture}
	\caption{Cylindrical Algebraic Decomposition (CAD) adapted to the unit disk.}
	\label{fig:cad_example}
\end{figure}
\begin{example}
	Here is a CAD of $\R^2$ adapted to the closed unit disk $\{(x_1,x_2)\mid x^2_1+x^2_2=1\}$ and the polynomial $x_1^2+x_2^2-1$. The 1-cells are given by $\mathcal{S}_1\leq \{(-\infty,-1),\{-1\},(-1,1),\{1\},(1,\infty)\}$. There is only one 2-cell above $(-\infty,-1)$ (\resp $(1,\infty)$), namely the cylinder $(-\infty,-1)\times\R$ (\resp $(1,\infty)\times \R$). Above $\{-1\}$ there are three 2-cells: $\{-1\}\times (-\infty, 0)$, $\{-1\}\times \{0\}$ and $\{-1\}\times (0,\infty)$ (and similarly above $\{1\}$)),  \ie $\xi_{(-1,1),1}=0$. Finally, above $(-1,1)$ there are five 2-cells defined by the two semi-algebraic functions $\xi_{(-1,1),1}(x_1)=-\sqrt{1-x_1^2}$ and $\xi_{(-1,1),1}(x_1)=\sqrt{1-x_1^2}$ representing the elements of $(-1,1)\times \R$ which are strictly above the unit disk, on the positive boundary of the unit disk, in the interior of the unit disk, on the negative boundary of the unit disk and strictly below the unit disk respectively. The closed unit disk is the union of 3 of the 2-cells, and the CAD is therefore adapted as claimed. The 1-cells and 2-cells of this example are represented in \cref{fig:cad_example}.
\end{example}

The problem of computing a CAD adapted to a finite collection of semi-algebraic sets or polynomials was initially solved by Collins in~\cite{collins1976} and later refined in \eg~\cite{arnon1984cylindrical,collins1991partial,strzebonski2000solving,basu2006algorithms}.
Algorithmically, CAD does not produce an explicit description of the semi-algebraic functions $\xi_{S,i}, 1\leq i\leq k_S$, for every cell $S$ in the decomposition (this would in any case be impossible since the roots of polynomials of degree $\geq 5$ do not in general have algebraic expressions). Instead, it relies on the fact that it is enough to find a single \emph{sample point} in each cell $S\in \mathcal{S}_i$ to be able to express these functions \emph{implicitly} as the roots of so-called `projection polynomials' $p:\R^{i+1}\to \R^i$ denoted $\mathrm{Root}_j(p(x_1, \ldots, x_i,\cdot ))$ and whose existence is guaranteed by the implicit function theorem and the definition of these polynomials, see~\cite{arnon1984cylindrical,basu2006algorithms}. Formally, given a finite family polynomials $P_1, \ldots, P_k$ in $n$ variables and of degree bounded by $d$, the CAD algorithm returns the tree structure of a CAD adapted to $P_1, \ldots, P_k$ and a (rational or algebraic) sample point in each cell~\cite{arnon1984cylindrical}. The complexity of this algorithm is $(kd)^{O(1)^{n-1}}$, \ie polynomial in the degree but doubly exponential in the number of variables \cite[\S 11]{basu2006algorithms}. 
In our experimental evaluation, we use Mathematica's implementation of CAD, \texttt{CylindricalDecomposition[]}~\cite{strzebonski2000solving} and the C\texttt{++} implementation QEPCAD~\cite{10.1145/968708.968710,collins1991partial} for which we have developed a Python library\footnote{A similar function is \texttt{Reduce[]}, which mainly uses  \texttt{CylindricalDecomposition[]}, but also uses other methods for quantifier elimination (virtual substitution, Gr\"{o}bner bases). We therefore choose to use \texttt{CylindricalDecomposition[]} directly.}.

\paragraph{Using CAD for arithmetic on RVs.}
Our aim is to compute the integral \eqref{eq:general_integral}  when $f(x_1, \ldots, x_n)=f_1(x_1)\ldots f_n(x_n)$, \ie when the $n$ input variables are independent, for an arbitrary arithmetic expression $\varphi(x_1, \ldots, x_n)$ (in particular variables can be repeated). To achieve this, we will construct a CAD $(\mathcal{S}_i)_{1\leq i\leq n}$ adapted to the semi-algebraic set 
\[
T=\{(t, x_1, \ldots,x_n)\in \R^{n+1}\mid \varphi(x_1, \ldots, x_n)-t< 0\}.
\]
Moreover, we will choose this CAD in such a way that the 1-cells describe the possible ranges of values taken by the variable $t$. Beyond this, we can assume w.l.o.g.\ that the projections follow the variable ordering $x_1, \ldots, x_n$, \ie the 2-cells describe the ranges of possible values for the variables $(t,x_1)$, etc. We will also assume that we can restrict our attention to the part of the CAD partitioning $T$, that is to say the cells below the set $\mathcal{S}_{n+1}$ of $(n+1)$-cells partitioning $T$.  

From the definition of a CAD, there are $2k_S+1$ cells of level $i+1$ above any given $i$-cell $S$, $1\leq i\leq n$. However, $k_S$ of these have dimensionallity $i$, being the graph a semi-algebraic function $\xi_{S,j}: S\to \R$. They will therefore not contribute to our integration problem \eqref{eq:general_integral} since all cells above them will necessarily have a dimension that is smaller than their ambient space and thus be of (Lebesgue) measure zero. We will therefore ignore these cells and only consider the $k_S+1$ bands above $S$ defined by the semi-algebraic maps $\xi_{S,1}, \ldots\xi_{S,k_S}$; this simplification is considered for the same reason in~\cite{strzebonski2000solving}.  Finally, to write down our expression for \eqref{eq:general_integral} we introduce the following convention. We will write $S_{i_1}, 0\leq i_1\leq k_1$ for the cells of level 1 that are open sets (and defined by $k_1-1$ points), we will write $S_{i_1, i_2}, 1\leq i_2\leq k_{S_{i_1}}$ for the (two-dimensional) 2-cells above $S_{i_1}$, and so on until $S_{i_1,\ldots,i_{n+1}}, 1\leq i_{n+1}\leq k_{S_{i_1,\ldots,i_{n}}}$ for the cells of level $n+1$ above the cell  $S_{i_1,\ldots,i_{n}}$. With this notation and a CAD adapted to the set $T$ we can write the CDF $F(t)$ in~\eqref{eq:general_integral} as the piecewise function defined at $t\in S_{i_1}$ by equation \eqref{eq:CAD_CDF} below. This expression illustrates the bad complexity of CAD \wrt the number of variables. However, we find in \cref{sec:experiments} that evaluating the CAD in \eqref{eq:CAD_CDF} is manageable in practice, the computational bottleneck being the symbolic evaluation of the integrals.
\begin{align}
	\hspace{-10ex }F(t) &=\int_{\{(x_1, \ldots,x_n)\mid \varphi(x_1, \ldots,x_n)\leq t\}}f_1(x_1)\ldots f_n(x_n)~dx_1\ldots dx_n \nonumber\\
	& = \hspace{-4ex}\sum_{\substack{0 \leq i_2\leq  k_{S_{i_1}} \\ 0\leq i_3\leq k_{S_{i_1,i_2}} \\ \cdots  \\ 0\leq i_{n+1}\leq k_{S_{i_1,\ldots,i_n}} }} 
	\hspace{-2ex}\int_{\xi_{S_{i_1},i_2}(t)}^{\xi_{S_{i_1},i_2+1}(t)}f_1(x_1)\int_{\xi_{S_{i_1,i_2},i_3}(t,x_1)}^{\xi_{S_{i_1,i_2},i_3+1}(t,x_1)}f_2(x_2)\ldots \nonumber
\end{align}
\vspace{-10ex}
\begin{align}
	&\hspace{22ex} \int_{\xi_{S_{i_1,\ldots,x_n},i_{n+1}}(t,x_1,\ldots,x_{n-1})}^{\xi_{S_{i_1,\ldots,i_n},i_{n+1}+1}(t,x_1,\ldots,x_{n-1})} f_n(x_n)~dx_1\ldots dx_n \label{eq:CAD_CDF}
\end{align}
	
\section{Stochastic Programs}\label{sec:programs}

\paragraph{Syntax.} 
We develop our exact/non-Monte-Carlo based evaluation for simple in-line probabilistic programs, specifically if-else statements, square roots and arithmetic expressions over independent continuous random variables. The grammar of programs that we support is given by:
\begin{align}
\text{Expressions:} & &	
e & ::= X \in \text{Var} \mid c \in \mathbb{R} \mid e + e \mid e - e \mid e \times e \mid e \div e \mid \sqrt{e} \nonumber\\
\text{Tests:} & & B & ::= e > 0 \mid \lnot B\mid B\wedge B \mid B\vee B \nonumber\\
\text{Programs:} & & p&::= e \mid \ite{B}{p}{p}\label{def:syntax}
\end{align}
where $\Var = \{X_1, X_2, ..., X_n\}$ is a finite and enumerated set of variables.

\paragraph{Semantics.}
We write $\AST(e)$ for the abstract syntax tree of $e$, $\Var(e)$ for the list of variables in $e$ ordered by the enumeration of $\Var$, $\Var(\AST(e))$ for the list of variables at the leaves of $\AST(e)$ ordered from left to right, and we define $M_e=\lvert \Var(e)\rvert$ and $N_e=\lvert\Var(\AST(e))\rvert$. Thus for $e=X_1\div(X_1+X_2)$, $\Var(e)=\{X_1, X_2\}$, $M_e=2$, $\Var(\AST(e))=(X_1, X_1, X_2)$ and $N_e=3$. We overload the syntactic operators $\{+,-,\times, \div, \sqrt{}\}$ to denote their usual interpretations on $\R$. 

\emph{Semantics of expressions}. We define both a deterministic semantics (used for tests) and a probabilistic semantics (used for programs) of an expression $e$ and denote them by $\dsem{e}$ and $\sem{e}$ respectively. We opt not to follow the familiar approach consisting in updating a global valuation function over all the variable names in the language. Instead, we define $\dsem{e}$ as a map $\R^{M_e}\to\R$ (\ie we only consider the variable in $e$ as inputs).  For this, we will need a bookkeeping function defined as follows: clearly $M_e\leq N_e$ and by combining the maps $\cp: \Var\to \Var\times \Var, X_i\mapsto(X_i, X_i)$ and $\sw: \Var\times \Var\to \Var\times \Var, (X_i,X_j)\mapsto (X_j,X_i)$ we can construct a map $\varphi_e: \Var^{M_e}\to \Var^{N_e}$ that turns the list $\Var(e)$ into the list $\Var(\AST(e))$. The same formal combination of $\cp$ and $\sw$ defines a map  $\varphi_e: \R^{M_e}\to \R^{N_e}$ that populates the leaves of $\AST(e)$ with $M_e$ concrete values, one for each variable in $\Var(e)$. 
We now define the semantics in two steps. First we inductively define $\dsem{\AST(e)}: \R^{N_e}\to\R$ as:
{\small
\begin{align*}
	&\text{Base cases:}& &\dsem{\AST(X_i)}:\R\to\R, x\mapsto x \qquad \dsem{\AST(c)}:\R^0\to\R, 0 \mapsto c \\
	&\text{Inductive case:}& &  \dsem{\AST(e_1\ast e_2)} = \ast\circ (\dsem{\AST(e_1)}\times \dsem{\AST(e_2)})  \quad\ast\in\{+, -, \times, \div\}\\
	& & & \dsem{\AST(\sqrt{e})} = \sqrt{}\circ \dsem{\AST(e)}
\end{align*}}
\hspace{-1ex}where $\R^0=\{0\}$, the 0-dimensional real vector space, and $\dsem{\AST(e_1)}\times \dsem{\AST(e_2)}:\R^{N_{e_1}}\times \R^{N_{e_2}} \simeq \R^{N_e} \to\R\times\R$ is the pairing of the maps $\dsem{\AST(e_1)}, \dsem{\AST(e_2)}$. We can now define
\[
	\dsem{e}=\dsem{\AST(e)}\circ \varphi_e
\]
where $\varphi_e$ is the bookkeeping map defined earlier which `prepares' the inputs of the function $\dsem{e}$ so that they can be fed to $\AST(e)$. This function will introduce perfect correlations/dependencies (when variable names are copied) in the probabilistic semantics that we define next. Thus, unless it is a permutation, $\varphi_e$ is the mathematical representation of why we cannot simply inductively evaluate the probabilistic semantics using the formulas in \eqref{eq:independent}.

The probabilistic semantics is defined in terms of the deterministic one via:
\begin{align}
	\sem{e}: \Prob(\R^{M_e})\to\Prob(\R), \Pro\mapsto \dsem{e}_\ast\Pro\label{eq:sem_expression}
\end{align}
where $\Prob(\R^k)$ is the space of probability measures on the Borel space $\R^{k}$, $\Pro$ is the input (joint) distribution over the variables in $\Var(e)$ and $\dsem{e}_\ast\Pro$ is the output distribution defined as the pushforward of $\Pro$ through $\dsem{e}$, i.e.\ $\dsem{e}_\ast\Pro(A)=\Pro(\dsem{e}\inv(A))$ for any Borel measurable $A\subseteq \R$. In practice, we will associate with each $X_i\in\Var(e)$ a probability distribution $\Pro_i\in\Prob(\R)$, and take the input distribution to be the corresponding product of these distributions (i.e.\ we will assume that the input random variables are independent). For example, we might be interested in computing a distribution like
\begin{align}
	\sem{X_1\div(X_1+X_2)}(\mathsf{Uniform}(0,1)\otimes\mathsf{Uniform}(2,3)) \label{eq:pushfwd_example}
\end{align}
In the language of random variables, this is the distribution of $X_1\div(X_1+X_2)$ when $X_1\sim \mathsf{Uniform}(0,1)$ and $X_2\sim \mathsf{Uniform}(2,3)$ are independent. By definition of the pushforward, the CDF of \eqref{eq:pushfwd_example} is
\begin{align*}
&\sem{X_1\div(X_1+X_2)}(\mathsf{Uniform}(0,1)\otimes\mathsf{Uniform}(2,3))\big((-\infty,t)\big) \\
=~ & \mathsf{Uniform}(0,1)\otimes\mathsf{Uniform}(2,3)\big(\{(x_1,x_2)\mid \dsem{X_1\div(X_1+X_2)}(x_1,x_2)<t\}\big)\\
=~& \int_{\left\{(x_1,x_2)\mid \frac{x_1}{x_1+x_2}<t\right\}} \one_{[0,1]}(x_1)\one_{[2,3]}(x_2)~dx_1dx_2
\end{align*}
which is of the shape \eqref{eq:general_integral}. CAD is the tool that will allow us to compute these integrals, and therefore the probabilistic semantics of expressions and programs.

\emph{Semantics of tests}.  The semantics of tests is defined inductively as expected
{\small \begin{align*}
	&\text{Base case:}  & &\dsem{e>0}=\{x\in\R^{M_e}\mid \dsem{e}(x)>0\}
	\\
	&\text{Inductive case:}& &\dsem{\neg B}=\dsem{B}^c, \dsem{B_1\wedge B_2}=\dsem{B_1}\cap\dsem{B_2}, \dsem{B_1\vee B_2}=\dsem{B_1}\cup\dsem{B_2}, 
\end{align*}}

\emph{Semantics of programs}. We need only define the semantics of conditional branching. We follow \cite{kozen1979semantics,dahlqvist2020semantics}: every test $B$ defines a measurable subset $\dsem{B}\subseteq \R^{M_B}$ and this subset defines a linear operator $z_B: \Prob_{\leq 1}(\R^{M_B})\to\Prob_{\leq 1}(\R^{M_B})$ -- where $\Prob_{\leq 1}(\R^{M_B})$ is the set of \emph{sub}distributions on $\R^{M_B}$ -- defined by $z_B(\Pro)(A)=\Pro(A\cap \dsem{B})$. Clearly, $z_B(\Pro)$ is in general a subdistribution since the total mass of $\R^{M_B}$ under $z_B(\Pro)$ is $z_B(\Pro)(\R^{M_B})=\Pro(\R^{M_B}\cap \dsem{B})=\Pro(\dsem{B})$. In terms of density functions, $z_B$ \emph{z}eroes out the PDF of $\Pro$ on $\dsem{B}^c$. Thus the density of $z_{X_1>0}(\mathsf{Normal}(0,1))$ for example is the positive half of the bell curve. We define
\begin{align}
	\sem{\ite{B}{p_1}{p_2}} = \sem{p_1}\circ z_{B}+ \sem{p_2}\circ z_{\neg B}\label{eq:sem_ite}
\end{align}
Intuitively,  $z_B$ and $z_{\neg B}$ split the input into a subdistribution over $M_B$-tuples that satisfes the guard and another that satisfies its complement; these are passed to their respective branches, then re-combined to produce the output distribution. 
	\section{Algorithms}\label{sec:algorithms}
In this section we describe the function $\Proc{CDF}{}$  from \cref{algo:CDF} that takes as input a program $p(X_1, \ldots, X_k)$ from the syntax \eqref{def:syntax} and for each $X_i, 1\leq i\leq k$, a PDF $f_i: (a_i,b_i)\to \R^+, a_i,b_i\in \R\cup\{-\infty,\infty\}$, and returns the CDF
\[
F_p(t) = \sem{p}(\Pro_1\otimes \ldots\otimes \Pro_k)\big((-\infty,t)\big).
\]

\paragraph{Semi-algebraic input set generation.} The first step consists in specifying the semi-algebraic set over which we will integrate the input PDF following \eqref{eq:general_integral}. This semi-algebraic set
will be defined from three types of symbolic inequalities:
\begin{enumerate}
	\item constraints on the support of the input variables, \viz $(a_i<X_i)\wedge (X_i<b_i)$,
	\item the guards of the if-else statements and their complements,
	\item the expressions $e(X_1, \ldots, X_n)<t$ at the leaves of the program's AST that specify the problem of computing the CDF.
\end{enumerate}

The symbolic semi-algebraic set associated with the input data described above is constructed recursively by the function \Proc{SAS}{} in \cref{algo:SAS} applied to the program $p(X_1, \ldots,X_n)$ and the symbolic variables range constraints
\begin{align}
	S=\bigwedge_{i=1}^k\big((a_i< X_i)\wedge (X_i<b_i)\big). \label{eq:initial_cond}
\end{align}

\begin{algorithm}[h!]
	\caption{Computing the Semi-Algebraic Set associated with the problem}
	{\small
		\begin{algorithmic}\label{algo:SAS}
			\Require{Program $p(X_1, \ldots, X_k)$, semi-algebraic set $S$ initialised as \eqref{eq:initial_cond}}
			\Ensure{The symbolic semi-algebraic set associated with computing $p$'s output CDF}
			\Function{SAS}{$p, S$}
			\If{$p=\ite{B}{p_1}{p_2}$} 
			\State $B \gets$ \Call{RemoveRoots}{$B$}
			\State ifSAS $\gets$ \Call{SAS}{$p_1$, $B\wedge S$}
			\State elseSAS $\gets$ \Call{SAS}{$p_2$, $\neg B \wedge S$}
			\State \Return (ifSAS $\vee$ elseSAS)
			\Else \Comment{Base case: $p(X_1, \ldots, X_n)$ is an expression}
			\State \Return $S\wedge$ \Call{RemoveRoots}{$p-t<0$}
			\EndIf
			\EndFunction
			\State
			\Require{A boolean combination $B$ of inequalities $e_1<0, \ldots, e_p<0$ in DNF}
			\Ensure{A boolean combination of root-free inequalities that define the same semi-algebraic set as $B$ in DNF}
			\Function{RemoveRoots}{$B$}
			\State output $\gets$ $\varnothing$
			\For{disjunct in $B$}
			\State conj $\gets$ $\varnothing$
			\For{conjunct in disjunct}
			\State $e_1, e_2, e_3$ $\gets$ \Call{CanonicalForm}{conjunct} 
			\State \Comment{Put inequalities in the form $(\sqrt{e_1}\times e_2)+e_3<0$, see \cref{prop:sqrt_free}}
			\State I1 $\gets$  \Call{RemoveRoots}{$e_1 > 0\wedge e_2>0\wedge e_1e_2^2<e_3^2$}
			\State I2 $\gets$ \Call{RemoveRoots}{$e_1 > 0\wedge e_2<0\wedge  e_1e_2^2>e_3^2$}
			\State conj $\gets$ conj $\wedge$  (I1 $\vee $ I2)
			\EndFor
			\State output $\gets$ output $\vee$ conj 
			\EndFor
			\State \Return \Call{DNF}{output}
			\EndFunction
		\end{algorithmic}	
	}
\end{algorithm}

The function $\Proc{SAS}{}$ of \cref{algo:SAS} calls a function to get rid of square root expressions which, strictly speaking, are not part of the syntax defining semi-algebraic sets. It is not hard to check that this can always be done as is proved below. Our syntax only includes square roots, but the result below can easily be generalised to all roots and we could therefore have included them in our syntax.
\begin{proposition}\label{prop:sqrt_free}
	For every inequality of the form $e(X_1, \ldots, X_k)<0$, where $e(X_1, \ldots, X_k)$ is constructed from \eqref{def:syntax}, $\dsem{e(X_1, \ldots, X_k)<0}$ is semi-algebraic.
\end{proposition}
\begin{proof}
	Using \cref{prop:division_free} we can assume w.l.o.g.\ that $e$ is subtraction- and division-free. By replacing all subterms $\sqrt{e_1\times e_2}$ by $\sqrt{e_1}\times \sqrt{e_2}$, repeatedly applying distributivity, and using commutativity and associativity we can assume that the first square root node in $\AST(e)$ (using preorder traversal) occurs in a tree of the shape $(\sqrt{e_1}\times e_2)+e_3$ (if $\sqrt{}$ is the root, take $e_2=1, e_3=0$, and similarly if it sits only below an addition, or only below a multiplication; if there is no root take $e_1=0$).
	We can now eliminate the square root in the arithmetically obvious way:
	\begin{align*}
		& \{x\in \R^k\mid (\sqrt{e_1(x)}\times e_2(x))+e_3(x)<0\}\\
		=~& \{x\in \R^k\mid e_1(x)>0\wedge e_2(x)>0\wedge e_3(x)<0\wedge e_1(x)\times (e_2(x))^2<(e_3(x))^2\}
		\\
		& \cup \{x\in \R^k\mid e_1(x)>0\wedge e_2(x)<0\wedge e_3(x)>0\wedge e_1(x)\times (e_2(x)^2>(e_3(x)
		)^2\}
	\end{align*}
	By recursively applying the procedure described above to these new inequalities if they contain square roots we eliminate all square roots from the initial expression and are left with boolean combinations of semi-algebraic sets, \viz a semi-algebraic set.\hfil\qed
\end{proof}

Given the semantics of if-else branching described in \eqref{eq:sem_ite}, it may not immediately be clear that aggregating the preconditions, the guards, their complements, and the inequalities $p-t<0$ at each leaf in the way described by \cref{algo:SAS} will produce the desired semantics. The following result shows that it does.
\begin{proposition}\label{prop:semantic_correctness}
	\cref{algo:SAS} is semantically correct in the sense that 
	\[
		\sem{p}(\mathbb{P}_1\otimes \ldots\otimes \mathbb{P}_k)\big((-\infty,t)\big)=\mathbb{P}_1\otimes \ldots\otimes \mathbb{P}_k\big(\dsem{\mathrm{SAS}(p,S)}\big)
	\]
\end{proposition}
\begin{proof}
	We show by induction on the structure of $p$ that 
	\[
	\sem{p}\circ z_B(\mathbb{P}_1\otimes \ldots\otimes \mathbb{P}_k)\big((-\infty, t)\big) = \mathbb{P}_1\otimes\ldots \otimes \mathbb{P}_k\big(\dsem{\mathrm{SAS}(p, B\wedge S)}\big)
	\]
	The result then follows by picking any tautology for $B$, for example $1>0$.
	In the base case, $p=e(X_1, \ldots,X_k)$ is an expression and we have
	{\small{
			\begin{align*}
				&\sem{e}\circ z_B(\mathbb{P}_1\otimes\ldots \otimes \mathbb{P}_k)\big((-\infty,t)\big)\\
				\stackrel{(i)}{=}~&\dsem{e}_\ast z_B(\mathbb{P}_1\otimes\ldots \otimes \mathbb{P}_k)\big((-\infty,t)\big)\\
				\stackrel{(ii)}{=}~&\mathbb{P}_1\otimes\ldots \otimes \mathbb{P}_k\big(\{(x_1, \ldots, x_k)\mid \dsem{e}(x_1, \ldots,x_k)<t \wedge (x_1,\ldots,x_k)\in \dsem{B}\big\})\\
				\stackrel{(iii)}{=}~&\mathbb{P}_1\otimes\ldots \otimes \mathbb{P}_k\big(\{(x_1, \ldots, x_k)\mid \dsem{e}(x_1, \ldots,x_k)-t<0 \wedge (x_1,\ldots,x_k)\in \dsem{B\wedge S}\}\big) \\
				\stackrel{(iv)}{=}~&\mathbb{P}_1\otimes\ldots \otimes \mathbb{P}_k\big(\dsem{\mathrm{SAS}(e, B\wedge S)}\big)
	\end{align*}}}where ($i$) is from \eqref{eq:sem_expression}, ($ii$)  follows from the definition of pushforward and of $z_B$, ($iii$) follows from the definition \eqref{eq:initial_cond} of $S$ (which means that $\dsem{S}$ contains the support of $\mathbb{P}_1\otimes\ldots \otimes \mathbb{P}_k$) and ($iv$) is by definition of \cref{algo:SAS}.
	
	For the inductive case we have $p=\ite{B'}{p_1}{p_2}$ and therefore
	{\small{
			\begin{align*}
				&\sem{\ite{B'}{p_1}{p_2}}\circ z_B(\mathbb{P}_1\otimes\ldots \otimes \mathbb{P}_k)\big((-\infty, t)\big)\\
				\stackrel{(i)}{=}~&\sem{p_1}\circ z_{B'}\circ z_B(\mathbb{P}_1\otimes\ldots \otimes \mathbb{P}_k)\big((-\infty,t)\big)+\sem{p_1}\circ z_{\neg B'}\circ z_B(\mathbb{P}_1\otimes\ldots \otimes \mathbb{P}_k)\big((-\infty,t)\big)\\
				\stackrel{(ii)}{=}~&\sem{p_1}\circ z_{B'\wedge B}(\mathbb{P}_1\otimes\ldots \otimes \mathbb{P}_k)\big((-\infty,t)\big)+\sem{p_1}\circ z_{\neg B'\wedge B}(\mathbb{P}_1\otimes\ldots \otimes \mathbb{P}_k)\big((-\infty,t)\big)\\
				\stackrel{(iii)}{=}~& \mathbb{P}_1\otimes\ldots \otimes \mathbb{P}_k\big(\dsem{\mathrm{SAS}(p_1, B'\wedge B\wedge S)}\big)+ \mathbb{P}_1\otimes\ldots \otimes \mathbb{P}_k\big(\dsem{\mathrm{SAS}(p_2, \neg B'\wedge B\wedge S)}\big)\\
				\stackrel{(iv)}{=}~& \mathbb{P}_1\otimes\ldots \otimes \mathbb{P}_k\big(\dsem{\mathrm{SAS}(p_1, B'\wedge B\wedge S)\vee \mathrm{SAS}(p_2, \neg B'\wedge B\wedge S)}\big) \\
				\stackrel{(v)}{=}~& \mathbb{P}_1\otimes\ldots \otimes \mathbb{P}_k\big(\dsem{\mathrm{SAS}(p, S)}\big) 
	\end{align*}}}where ($i$) is by \eqref{eq:sem_ite}, ($ii$) is an immediate consequence of the definition $z_B$, ($iii$) is by the induction hypothesis, ($iv$) is by the fact that $\dsem{B'\wedge B\wedge S}$ and $\dsem{\neg B'\wedge B\wedge S}$ are disjoint and the additivity of measures, finally ($v$) follows from the definition of \cref{algo:SAS}.\hfil\qed
\end{proof}

\paragraph{CAD.} Once the semi-algebraic input set $S$ has been constructed by $\Proc{SAS}{}$, it is fed to $\Proc{CylindricalAlgebraicDecomposition}{S, \{t, X_1, \ldots, X_n\})}$ that will construct the associated CAD, using the specified variable ordering $\{t, X_1, \ldots, X_n\}$. $\Proc{CylindricalAlgebraicDecomposition}{}$ returns a boolean combination of inequalities that we assign to the variable CADTree in \cref{algo:CDF}. Its exact shape is in practice somewhat  unpredictable and putting it in DNF allows us to traverse it conveniently. Each disjunct of $\Proc{DNF}{\mathrm{CADTree}}$ corresponding to a given 1-cell $(c_i,d_i)$ will be a conjunction of the shape 
\[
c_i< t<d_i\wedge f_1(t)<x_1<g_1(t)\wedge\ldots\wedge f_k(t, x_1, \ldots, x_{k-1})<x_k<f_k(t, x_1, \ldots, x_{k-1}).
\]
As discussed in \cref{sec:prelim}, the semi-algebraic functions $f_i(t,x_1, \ldots,x_{i-1}), 1\leq i\leq k-1$ in the expression above are defined as the parametrised roots of a (projection) polynomial in $x_{i}$. Unless this root is expressible explicitly by radicals, it will be represented implicitly in this way. 

\paragraph{Full and semi-evaluations of CDF}. Each of the disjuncts above allow us to compute one of the summands in the CAD-based expression of the CDF~\eqref{eq:CAD_CDF}. This computation, \viz piece $\leftarrow$ piece $+~ \Proc{Integration}{\mathrm{conjunction}, f_1, \ldots,f_k}$, defines the CDF as a piecewise function (the guard of the \textbf{if} just above checks whether the conjunction contains the current one-cell $c_i< t<d_i$.). In implementing $\Proc{Integration}{}$ in  \cref{algo:CDF} there are three options: we can perform it  symbolically, numerically at specific values of $t$, or as a mixture of the two, with some summands of~\eqref{eq:CAD_CDF} evaluated symbolically and others left unevaluated -- and thus requiring numerical evaluation -- using a best effort strategy. In the first case we will talk of full evaluation, in the other cases of semi-evaluation\footnote{Thus $ \frac{t^3}{3} + \frac{t^2}{2}$ is a full evaluation of $\int_{0}^{t} (x^2 + x)~dx$ and $\int_{0}^{t} x^2~dx+\frac{t^2}{2}$ is semi-evaluated.}.

\begin{algorithm}[hbt]
	\caption{Algorithm to compute the CDF of $p(X_1, ..., X_k)$}
	\label{algo:CDF}
	{\small
		\begin{algorithmic}
			\Require{Program $p(X_1, \ldots, X_k)$, PDFs $f_1, \ldots , f_k$ with ranges $(a_1,b_1), \ldots(a_k,b_k)$.}
			\Ensure{Piecewise CDF function $F$}
			\Function{CDF}{$p, S, f$}
			\State F $\gets$ [] \Comment{Empty piecewise function}
			\State InputSAS $\gets$ SAS($p,\bigwedge_{i=1}^k (a_i<X_i)\wedge (X_i<b_i)$)
			\State CADTree $\gets$ \Call{CAD}{InputSAS, $\{t, X_1, \ldots, X_n\}$)}
			\State \Comment{Optional: try different variable orderings $\{t, X_{i_1}, \ldots, X_{i_n}\}$ to get fewer cells}
			\State CADTree $\gets$ \Call{RemoveMeasureZeroCells}{CADTree}
			\State OneCells $\gets$ \Call{GetOneCells}{CADTree} \Comment{Defines the piecewise structure F}
			\State CADTree $\gets$ \Call{DNF}{CADTree}
			\For{cell in OneCells}
			\State piece $\gets$ 0
			\For{conjunction in CADTree}
			\If{conjunction $\subset$ cell}
			\State piece $\gets$ piece + \Call{Integration}{conjunction, $f_1\cdots f_k$}
			\State \Comment{Optional: run integrations in parallel}
			\EndIf
			\EndFor
			\State F.append([cell, piece])
			\EndFor
			\State \Return F \Comment{The CDF is then plotted and differentiated to get the PDF}
			\EndFunction
		\end{algorithmic}
	}
\end{algorithm}

\paragraph{Performance optimisation.} 
We have seen that CAD is doubly exponential in the number of variables in the program. In practice however, we observe that the computational bottleneck is not the generation of the CAD, but the nested integrals of~\eqref{eq:CAD_CDF} that are exponential in the number of variables. To speed up this part of \cref{algo:CDF} we use two strategies. The first is to find a good variable elimination ordering as this can lead to significantly fewer cells, thus reducing the cost of integration.
Brown's heuristic~\cite{Brown2004Tutorial,delRio2022Heuristic} aims to achieve this by ordering variables from the simplest (low maximal degree and few occurrences) to the most complex (high maximal degree and many occurrences).
Second, since every $\Proc{Integration}{\mathrm{conjuction}, f_1,\ldots,f_k}$ computation is independent of the others, we can parallelise the computation of the integrals.

\paragraph{Implementations.} We have developed two implementations of \cref{algo:CDF}. The first, written in Mathematica, is designed to compute full evaluations and relies on Mathematica's powerful symbolic integration capabilities. It provides our gold-standard solutions. The second, using only open-source code, aims to efficiently obtain semi-evaluated solutions that can be evaluated or plotted fast. Both tools implement the performance optimisations described above.

\emph{Mathematica implementation.} Mathematica's \texttt{CylindricalDecomposition[]} function~\cite{strzebonski2000solving} computes our CADs. Mathematica encodes the root expressions defining a CAD by using $\mathtt{Root}[p(x_1, \ldots x_i,\#)\&,j]$ expressions, standing for the $j^{th}$ real root of the (projection) polynomial $p(x_1, \ldots x_i,x_{i+1})$ in the variable $x_{i+1}$. Its powerful integration engine allows us to compute closed-form solution for nested integrals containing root expressions in their bounds. This allows us to evaluate \eqref{eq:CAD_CDF} in closed-form. Mathematica can also differentiate these expressions and produce closed-form expressions for the corresponding PDFs.

\emph{Open-source implementation.} We have developed Py-CAD, a tool written in Python and C\texttt{++} that performs the (semi-)evaluation of \eqref{eq:CAD_CDF}.  After denominator and radical elimination (\cref{prop:division_free,prop:sqrt_free}) Py-CAD starts by simplifying the input semi-algebraic set $S$ as much as possible using Z3~\cite{z3}. The simplified input is fed to QEPCAD~\cite{10.1145/968708.968710,collins1991partial}. QEPCAD produces projection polynomials and constructs cells consisting of an index and a sample point. Py-CAD parses the projection polynomials into SymPy expressions and, for each cell, reconstructs the symbolic boundaries as SymPy \texttt{RootOf} objects by solving the projection polynomials at the corresponding cell sample points.
These cells are then assembled into a disjunctive normal form representation that is logically-equivalent to that produced by Mathematica. Py-CAD can then assemble the CDF~\eqref{eq:CAD_CDF} from integrals with exact algebraic boundaries regardless of the existence of closed form radical equivalents. During evaluation, a pass is conducted to identify fully determined \texttt{RootOf} nodes and substitute them for their numerical value via root-finding, allowing the integration engine to proceed. Py-CAD (like PSI) uses symbolic integration when possible  (\eg via C\texttt{++} pattern matching) and then performs numerical evaluation on semi-evaluated expressions. In order to improve performance on larger benchmarks, we introduce compositional sub-CADs: when an independent subtree of variables is identified (e.g. $y^2 + z^2$ within $x^2 + y^2 + z^2 < t$), its distribution is computed via an isolated sub-CAD and injected as a PDF into the outer problem, reducing a three-variable CAD to two two-variable CADs, a significant saving given that the complexity is doubly exponential in the number of variables~\cite{basu2006algorithms}. 
For now, Py-CAD is focusing on numerical integration of semi-evaluated expressions, with future plans to include full symbolic integration.
	\section{Experimental Evaluation}\label{sec:experiments}

\paragraph{Experimental setup.} We test both our implementations on programs and input distributions of various levels of complexity, from Taylor expansions (of order 5) of well-known functions, to programs from the FPBench benchmark suite~\cite{fpbench} used in floating-point analysis, to three functions from IMU libraries (including \texttt{getMeasurement} from \cref{sec:intro}). Run times and various metrics are presented in \cref{tab:benchmarks}. The column $t_{math}$ (\resp $\hat{t}_{math}$) shows the run time for our Mathematica full evaluation (\resp numerical evaluation over 400 equidistant points), $\hat{t}_{open}$ show the run time of Py-CAD's semi-evaluation over 400 equidistant points and $t_{psi}$ that of PSI's full evaluation (completed by Mathematica if marked by *).

The tools (artifact available at https://doi.org/10.5281/zenodo.20080511) take as input a program as well as a description of its input distributions. When a range of values is already defined in the literature, \eg for the FPBench benchmarks, we stick to this range and place a distribution over it. We choose as input distributions on ranges $[a,b]$ either the uniform $\unif{a}{b}$, the triangular $\tri{a}{b}$ or the $\mathrm{Beta}(2,2,a,b)$ distributions. The latter have PDFs
\begin{align*}
	f_\mathrm{T}(x) = \begin{cases}
		\frac{4(x - a)}{(b-a)^2} & \text{if } a < x \le a+\nicefrac{(a+b)}{2} \\
		\frac{4(b - x)}{(b-a)^2} & \text{if } a+\nicefrac{(a+b)}{2} < x \le b.
	\end{cases}
	&~ &
	f_{\mathrm{Beta}}(x)=\one_{[a,b]}(x)\frac{6(x-a)(b-x)}{(b-a)^3}
\end{align*}
Distributions with infinite or compact supports are supported by our tool, the only practical limitation is the tractability of their PDFs in symbolic integrations. 

\begin{table}[htbp]
\begin{center}
		\scriptsize{
\caption{\footnotesize{Running time and KS-distance of a collection of benchmarks. The tools are run on an 8-core laptop with Apple M2 chip and 16GB RAM, or 10-core AMD Ryzen AI 9 with 32GB RAM (Py-CAD). The code of the benchmarks can be found in the Appendix together with plots of the output CDFs and PDFs. Depth is the Mathematica tree-depth of the expression (this is smaller than the tree depth given by our syntax \eqref{def:syntax}). \# vars describes the number of variables in the program, \# cells is the total number of cells in the corresponding CAD -- and therefore a good proxy for the complexity of the symbolic integration. KS denotes the KS distance between the symbolic distribution and the empirical distribution based on $10^6$ input samples. Timings are given as follows. $t_{mc}$ for the Monte-Carlo approach ($10^6$ samples). $t_{math}$ for fully-evaluated Mathematica, while $\hat{t}_{math}$ is numerically-evaluated Mathematica. $\hat{t}_{open}$ is from semi-evaluated Py-CAD. Finally, $t_{psi}$ is the timing for PSI. An asterix (*) is used where the PSI output contained integrals, which were then symbolically integrated via Mathematica (in parallel, and included in the timing). T-O stands for `timed out'; timeouts were chosen according to benchmark complexity. X means error (no result given).}}\label{tab:benchmarks}}
\begin{tabular}{|c|c|c|c|c|c|c|c|c|c|}
	\hline
	\textbf{benchmark} & \textbf{depth} & \textbf{\#vars} & \textbf{\#cells} & \textbf{KS} $\times 10^{-3}$ & $\mathbf{t}_{mc}$ & $\mathbf{t}_{math}$ & $\mathbf{\hat{t}}_{math}$   & $\mathbf{\hat{t}}_{open}$ & $\mathbf{t}_{psi}$ \\
	\hline\hline
	\multicolumn{10}{|l|}{1a. Taylor expansions - uniform inputs} \\
	\hline \hline
	exp & 3 & 1 & 2 & 0.97 & 1.69 & 0.86 & 1.19 & 2.85 & 5.49* \\
	\hline
	sin & 3 & 1 & 2 & 1.11 & 1.74 & 0.65 & 0.90 & 2.82 & 0.76*  \\
	\hline
	cos & 3 & 1 & 2 & 1.81 & 1.71 & 0.89 & 1.04 & 2.19 & 0.55* \\
	\hline
	log & 4 & 1 & 2 & 0.71 & 1.70 & 0.61 & 0.87 & 2.75 & 0.81*\\
	\hline \hline
	\multicolumn{10}{|l|}{1b. Taylor expansions - triangular inputs (timeout=5mins)} \\
	\hline \hline
	\text{exp} & 3 & 1 & 2 & 0.72 & 1.74 & 1.52 & 1.14 & 3.84 &  175.74* \\
	\hline
	\text{sin} & 3 & 1 & 2 & 0.54 & 1.73 & 1.03 & 0.87 & 4.08 & 0.76* \\
	\hline
	\text{cos} & 3 & 1 & 2 & 0.76 & 1.73 & 1.29 & 0.83 & 4.09 & T-O \\
	\hline
	\text{log} & 4 & 1 & 2 & 0.62 & 1.74 & 1.03 & 1.10 & 3.83 & 163.79* \\
	\hline \hline
	\multicolumn{10}{|l|}{2a. FPBench~\cite{fpbench} benchmarks - uniform inputs (timeout=20mins)} \\
	\hline \hline
	\text{sum} & 1 & 3 & 9 & 0.47 & 2.00 & 8.36 & 12.90 & 6.25 & 4.89\\
	\hline
	\text{xbyxy} & 3 & 2 & 7 & 0.69 & 1.76 & 1.10 & 1.37& 2.88 & X \\
	\hline
	\text{nonlin1} & 3 & 1 & 2 & 0.84 & 5.16 & 2.73 & 2.83& 1.95 & 1.58* \\
	\hline
	\text{bspline3} & 2 & 1 & 2 & 0.76 & 1.72 & 0.72 & 0.99& 2.08 & 0.25 \\
	\hline
	\text{cav10} & 4 & 1 & 5 & 0.67 & 2.92 & 0.82 & 1.16 & 2.89 & 0.72 \\
	\hline
	doppler1 & 4 & 3 & 14 & 0.91 & 1.81 & 27.22 & 53.04& 9.33 & 362.18*\\
	\hline
	doppler2 & 4 & 3 & 14 & 0.79 & 1.77 & 29.11 & 35.00 & 9.85 & T-O\\
	\hline
	doppler3 & 4 & 3 & 14 & 0.65 & 1.77 & 29.20 & 39.76 & 9.61 & T-O\\
	\hline
	rigidBody1 & 2 & 3 & 40 & 0.70 & 1.90 & 63.88 & X & 112.4 & T-O\\
	\hline \hline
	\multicolumn{10}{|l|}{2b. FPBench~\cite{fpbench} benchmarks - triangular inputs (timeout=20mins)} \\
	\hline \hline
	\text{sum} & 1 & 3 & 9 & 0.73 & 2.07 & 79.01 & 54.72 & 294.52 & 9.81 \\
	\hline
	\text{xbyxy} & 3 & 2 & 7 & 0.47 & 1.84 & 13.50 & 2.87 & 19.33 & T-O \\
	\hline
	\text{nonlin1} & 3 & 1 & 2 & 0.80 & 6.96 & 1.75 & 1.67 & 3.62 & 11.82* \\
	\hline
	\text{bspline3} & 2 & 1 & 2 & 0.49 & 1.74 & 1.19 & 1.00 & 3.21 & 0.28 \\
	\hline
	\text{cav10} & 4 & 1 & 5 & 1.09 & 2.99 & 0.99 & 1.24 & 5.51 & 0.87  \\
	\hline
	\text{doppler1} & 4 & 3 & 14 & 0.75 & 1.76 & T-O & 168.03 & 196.77 & T-O\\
	\hline
	\text{doppler2} & 4 & 3 & 14 & 0.96 & 1.83 & T-O & 208.13 & 220.26 & T-O \\
	\hline
	\text{doppler3} & 4 & 3 & 14 & 0.51 & 1.82 & T-O & 94.25 & 390.19 & T-O \\
	\hline \hline
	\multicolumn{10}{|l|}{2b. FPBench~\cite{fpbench} benchmarks - Beta(2,2) inputs (timeout=20mins)} \\
	\hline \hline
	\text{sum} & 1 & 3 & 9 & 0.76 & 2.13 & 45.89 & 105.34 & 88.14 & 8.83 \\
	\hline
	\text{xbyxy} & 3 & 2 & 7 & 1.04 & 1.90 & 3.09 & 1.85 & 7.8 & X \\
	\hline
	\text{nonlin1} & 3 & 1 & 2 & 1.26 & 6.59 & 1.90 & 1.68 & 2.45 & 4.87*\\
	\hline
	\text{bspline3} & 2 & 1 & 2 & 0.73 & 1.77 & 1.02 & 0.86 & 2.12 & 0.26 \\
	\hline
	\text{cav10} & 4 & 1 & 5 & 0.62 & 3.08 & 1.16 & 1.20 & 3.18 & 0.90 \\
	\hline
	\text{doppler1} & 4 & 3 & 14 & 0.92 & 1.89 & T-O & 360.76 & 41.86 & X\\
	\hline
	\text{doppler2} & 4 & 3 & 14 & 1.12 & 1.92 & T-O & 489.75 & 49.2 & X \\
	\hline
	\text{doppler3} & 4 & 3 & 14 & 0.63 & 1.95 & T-O & 260.12 & 60.86 & T-O \\
	\hline \hline
	\multicolumn{10}{|l|}{3a. IMU benchmarks - uniform inputs (timeout=60mins\&100mins)} \\
	\hline \hline
	v\_norm & 4 & 3 & 10 & 0.76 & 1.86 & 134.63 & 9.72 & 14.26 & T-O\\
	\hline
	change\_gyro & 3 & 3 & 9 & 0.78 & 7.71 & T-O & 230.93 & X & T-O\\
	\hline
	get\_meas & 4 & 3 & 40 & 1.57 & 5.42 & T-O & T-O &  5653.12& T-O\\
	\hline
\end{tabular}
\end{center}
\vspace{-2mm}
\end{table}

\paragraph{Evaluation of correctness.} 
To sanity-check our method, we compare its result to Monte-Carlo sampling, a fool-proof approximation of the ground truth. For this we draw $10^6$ samples from each distribution to create an empirical CDF and a histogram of the output distribution. The running time for this procedure is reported in column $t_{mc}$. We measure the distance between our full evaluation and Monte-Carlo sampling using the KS-distance
\[
KS(\mathbb{P}_{full}, \mathbb{P}_{mc}) = \sup_{t} \lvert F_{full}(t) - F_{mc}(t) \rvert 
\]
where $\mathbb{P}_{full}$ (\resp $F_{full}$) is our output distribution (\resp closed-form CDF) and similarly for Monte-Carlo
\footnote{We use Mathematica's built-in \texttt{KolmogorovSmirnovTest} function. When it takes too long (over  $30$ seconds), we estimate the KS distance by evaluating the CDFs over 1000 equidistant percentiles: $\max(_{1\leq i\leq 1000} \lvert (F_a(x_i) - F_{mc}(x_i) \rvert ) $}
. Examples of comparison between exact and empirical PDFs and CDFs can be found in \cref{tab:cdf_pdf_sample} with a complete list of graphs in the Appendix showing that our tools return the correct output distributions.
\begin{table}[h]
	\scriptsize{
		\begin{center}
			\caption{\small{Some examples of PDF and CDF.}}
			\label{tab:cdf_pdf_sample}
			\begin{tabular}{c  c}
				\hline
				CDF & PDF\\
				\hline
				\includegraphics[width=0.45\textwidth]{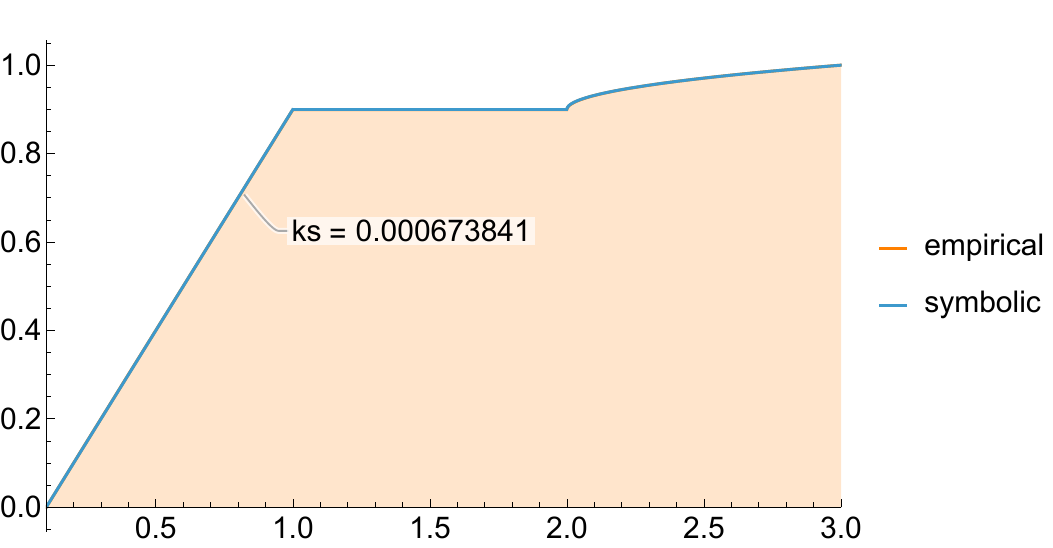} &
				\includegraphics[width=0.45\textwidth]{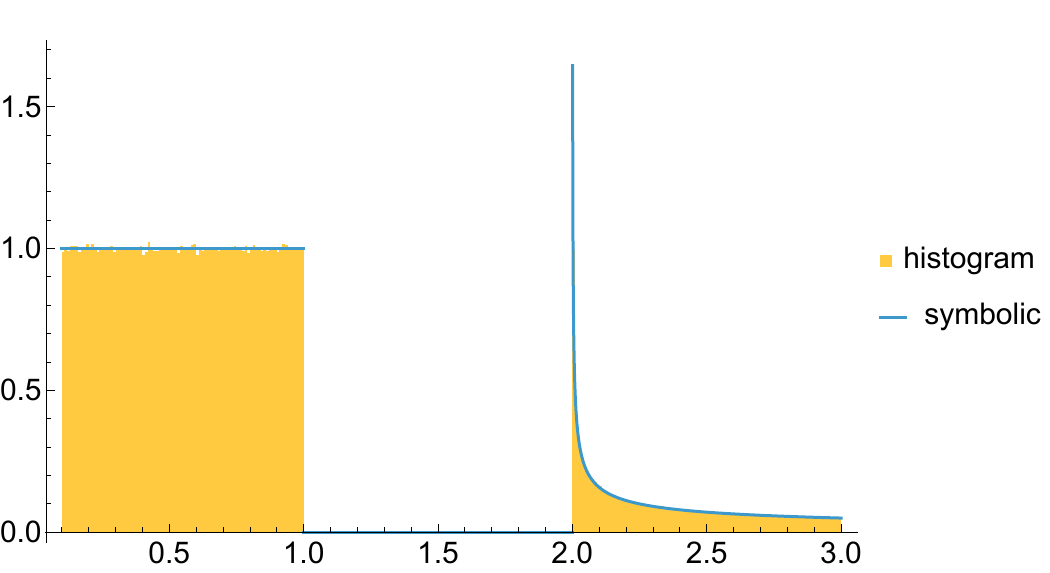} \\
				\hline
				\multicolumn{2}{c}{cav10: $\text{If}\left[x^2-x\geq 0,\frac{x}{10},x^2+2\right], x\sim\unif{0}{10}$} \\
				\hline
				\includegraphics[width=0.45\textwidth]{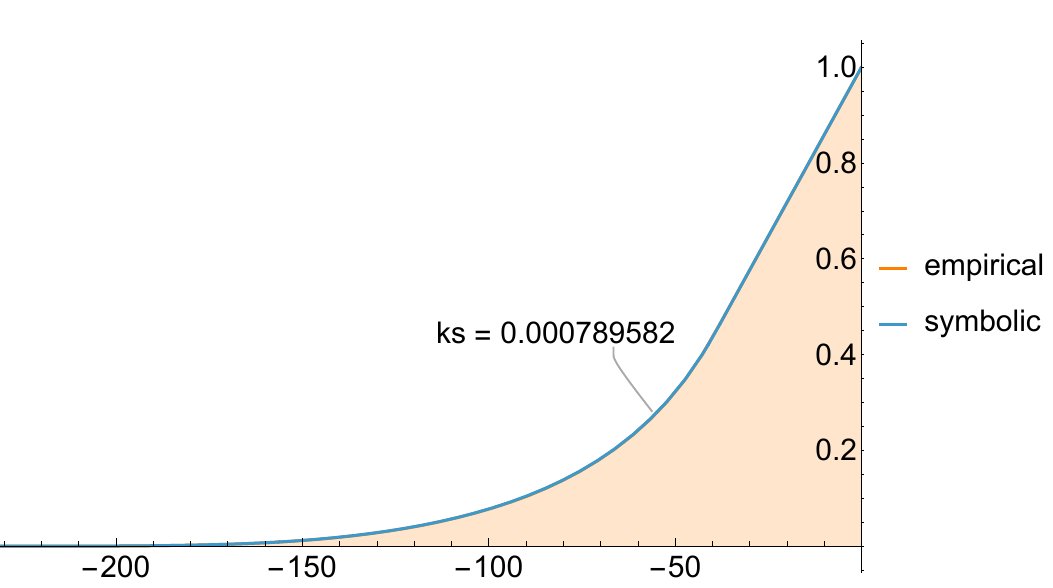} &
				\includegraphics[width=0.45\textwidth]{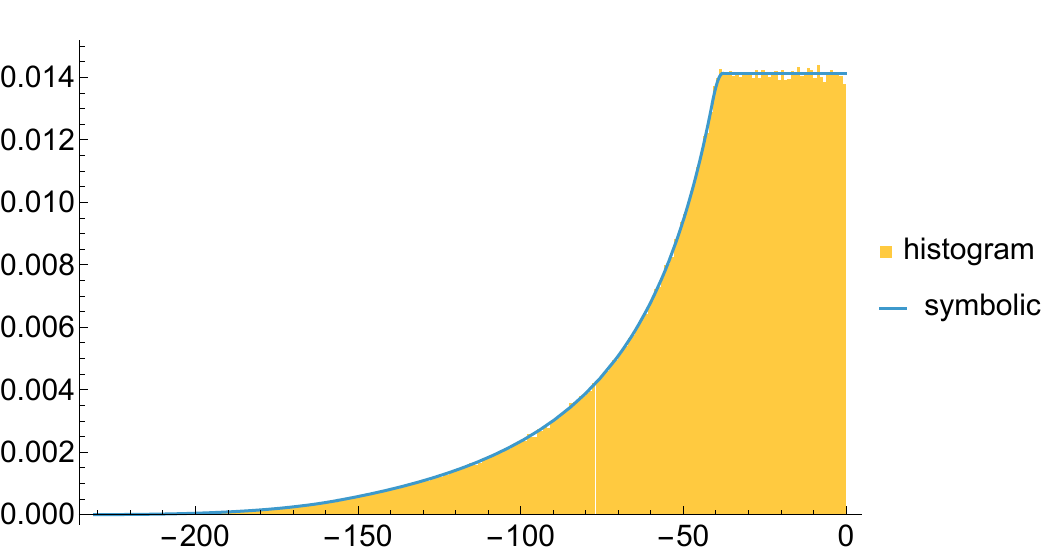} \\
				\hline
				\multicolumn{2}{c}{doppler2: $\frac{\left(-\frac{3 T}{5}-\frac{1657}{5}\right) v}{\left(\frac{3 T}{5}+u+\frac{1657}{5}\right)^2}, u\sim\unif{-125}{125}, v\sim\unif{15}{25000}, T\sim\unif{-40}{60}$} \\
				\hline \\
				\multicolumn{2}{c}{\includegraphics[width=0.95\textwidth,height=0.21\textheight, trim=5 12 12 8, clip]{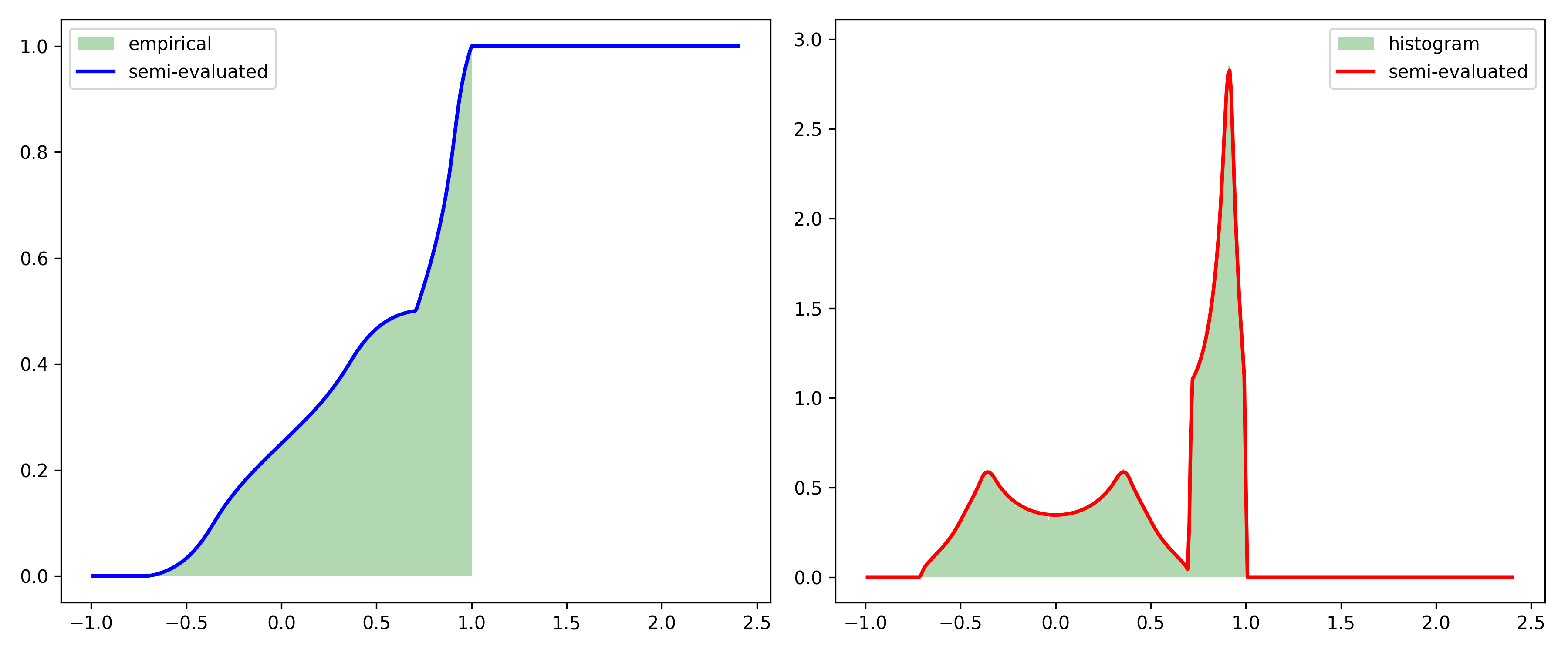}} \\
				\hline
				\multicolumn{2}{c}{get\_meas: c.f. \cref{sec:intro}, $x\sim\unif{-1}{1}, y\sim\unif{-1}{1}, z\sim\unif{-1}{1}$ [\textit{evaluated by Py-CAD}]} \\ 
				\hline
			\end{tabular}
		\end{center}	
	}
\end{table}

\paragraph{Discussion of results.}
Our Mathematica implementation produces symbolic CDFs and PDFs (see~\cref{tab:cdf_pdf_sample} and Appendix) that quickly become extremely complex. For some of the benchmarks in~\cref{tab:benchmarks}, \eg the \texttt{doppler} programs, the symbolic expressions returned by Mathematica are so large that they cannot be displayed on an A4 page. It is remarkable that CAD and symbolic integration are capable of evaluating these closed-form outputs in a reasonable time, providing a complete probabilistic range analysis for these programs, and in particular the ability to compute the probability of outcomes to arbitrary precision. 

It is interesting to note that the full evaluation is often competitive with the semi-evaluated numerical answer, \ie it is often more efficient to solve symbolically once than perform 400 numerical integrations. The running time in either case grows quickly as we move from one- or two-variable programs to three-variable ones. However, it remains doable to semi-evaluate the output distribution of a relatively complex program like get\_meas from \cref{sec:intro} at hundreds of points. In fact, in order to provide formal guarantees like $\Pro[\mathsf{BadOutcome}]<\varepsilon$, it is typically sufficient to evaluate the CDF at a handful of points (describing the region $\mathsf{BadOutcome}$), which will be orders of magnitude faster than the 400 evaluations of the CDFs and PDFs in our experiments. We can therefore expect that computing these kinds of exact probabilistic bounds should scale beyond the examples in our experiment.

Finally, we note that our tools perform well compared with PSI, which although much more general and able to perform exact inference, sometimes crashes and often times-out due to the complexity of the expressions it returns.

	\vspace{-4mm}
\section{Conclusions and Future Work}

We have presented a generic method for exactly evaluating the 
output distribution of programs with random inputs described by
density functions. This method relies on Cylindrical Algebraic Decomposition
and symbolic integration and handles if-else programs 
performing arithmetic operations and radicals.  We have implemented our method
in Mathematica and as Py-CAD, an open-source Python/C\texttt{++} tool, and applied it successfully to small benchmarks.

Despite its poor worse-case complexity, the main computational 
bottleneck of this approach is not the computation of the CAD, 
but the symbolic evaluation of nested integrals. Computing exact output distributions will always scale badly for this reason. However, we have shown that it is a realistic proposition for small programs from industrial application (FPBench) and sensor libraries. Using further refinements and optimisations (such as evaluating independent sub-programs using sub-CADs), we believe that we can further increase the range of programs for which exact output distributions can be either fully or semi-evaluated. Bearing in mind that for the purpose of program verification it is enough to be able to compute
the probability of certain outcomes, and that this typically only requires a few evaluations of the output CDF, we are hoping to apply our technique to larger programs, \eg four- or five-variable programs manipulating quaternions in IMU libraries.

Apart from these scalability questions, our next task will be to incorporate
floating-point errors to our framework. The model of 
\cite{constantinides2021rigorous,constantinides2025automated} seems
particularly well-suited to this task as the error is represented
as an independent noise variable applied multiplicatively 
at each arithmetic operation.

	\begin{credits}
		\subsubsection{\ackname} This work was supported by the Research Institute on Verified Trustworthy Software  (VeTSS)
		grant ``\emph{Probabilistic Precision Tuning}''.
	\end{credits}
	%
	%
	%
	%
	\bibliographystyle{splncs04}
	\bibliography{../mylibrary}
	\newpage
\section*{Appendix}\label{sec:appendix}	

	\scriptsize{
	\centering
	\captionof{table}{CDF and PDF of benchmarks vs.\ empirical distribution ($10^6$ samples)}
	\label{tab:cdf_pdf}
	\begin{longtable}{cc}
	
		\multicolumn{2}{c}{\textbf{\small Taylor approximations with uniform inputs}}\\
		\hline
		\\
		CDF & PDF \\
		\includegraphics[width=0.50\textwidth]{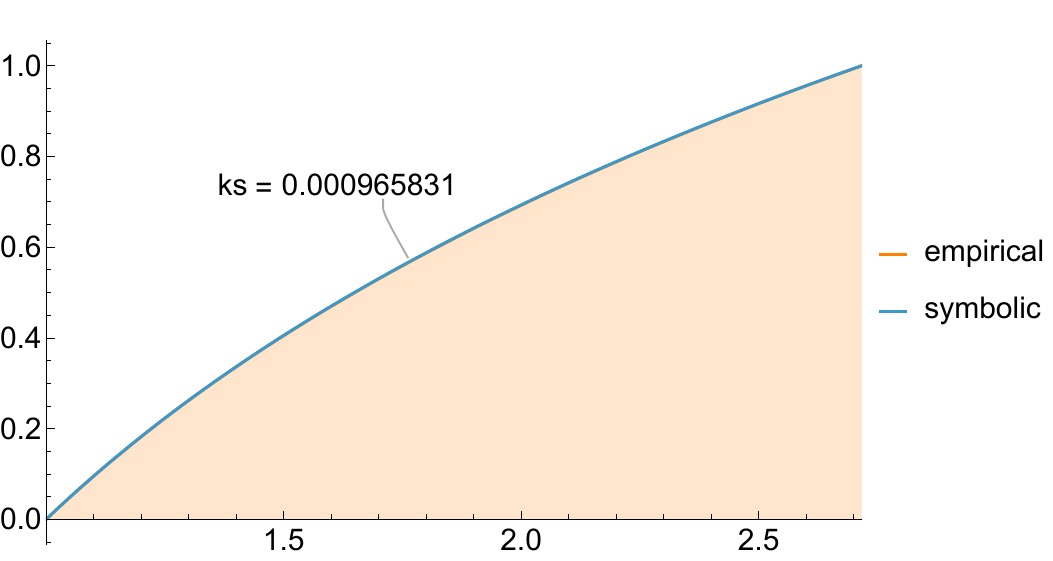} &
		\includegraphics[width=0.5\textwidth]{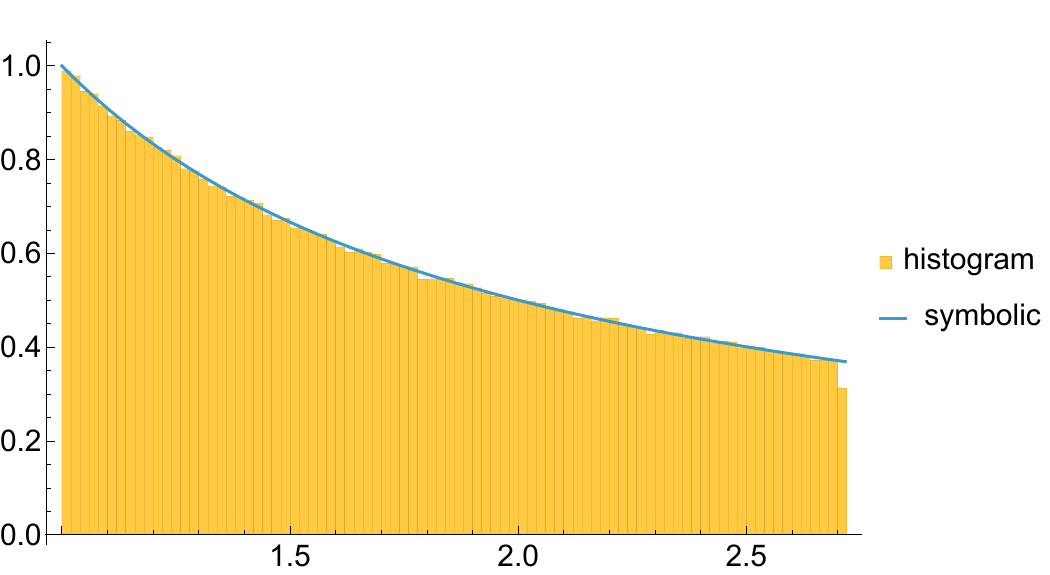} \\
		\hline
		\multicolumn{2}{c}{\rule[-0.75em]{0pt}{2em} exp: $1 + x + \frac{x^2}{2} + \frac{x^3}{6} + \frac{x^4}{24} +\frac{x^5}{120}, x\sim\unif{0}{1}$}\\
		\hline
		
		\includegraphics[width=0.50\textwidth]{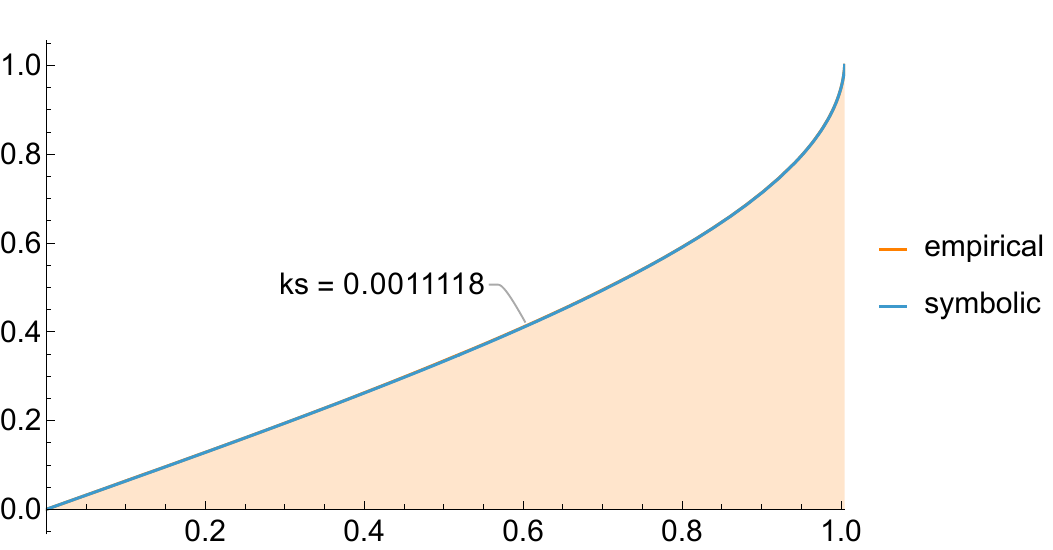} &
		\includegraphics[width=0.5\textwidth]{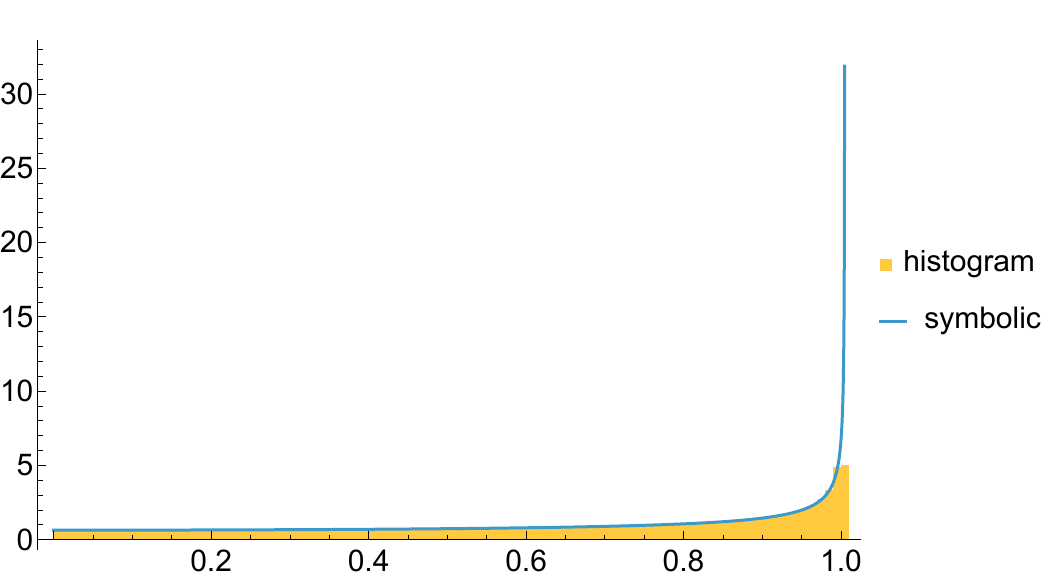} \\
		\hline
		\multicolumn{2}{c}{\rule[-0.75em]{0pt}{2em} sin: $\frac{x^5}{120}-\frac{x^3}{6}+x, x\sim\unif{0}{\nicefrac{\pi}{2}}$}\\
		\hline
		
		\includegraphics[width=0.50\textwidth]{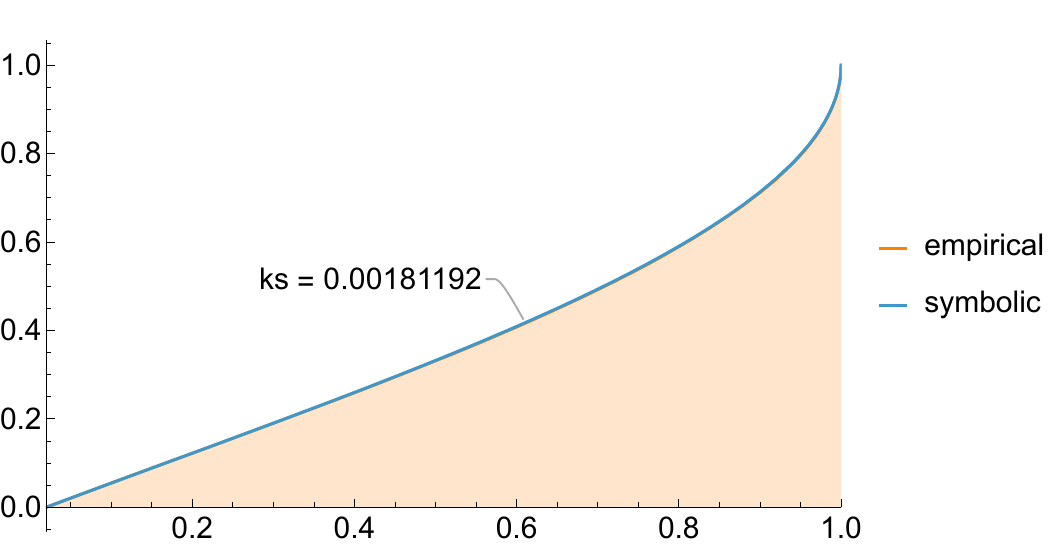} &
		\includegraphics[width=0.5\textwidth]{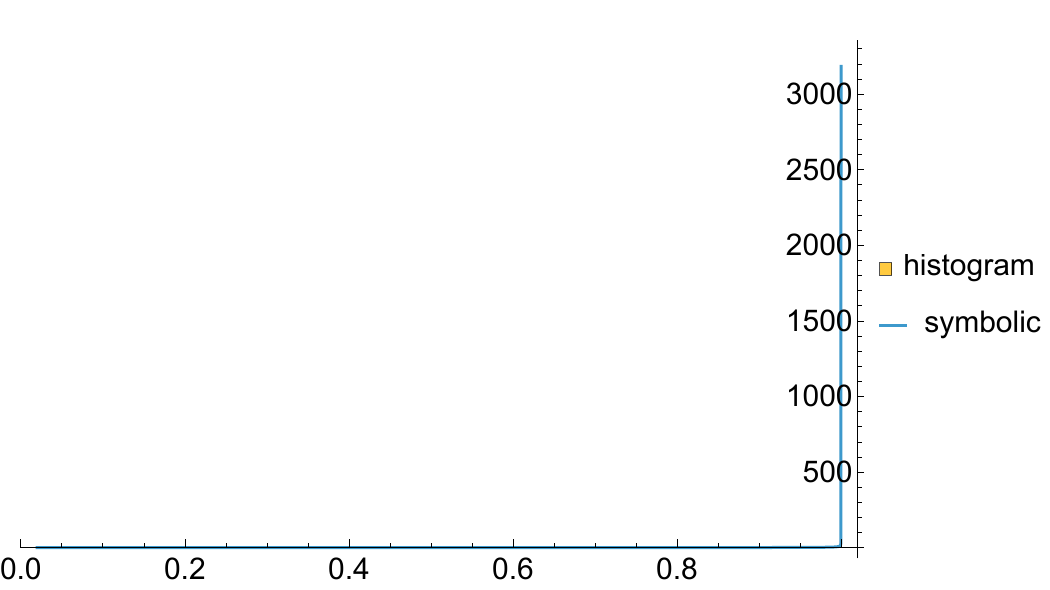} \\
		\hline
		\multicolumn{2}{c}{\rule[-0.75em]{0pt}{2em} cos: $\frac{x^4}{24}-\frac{x^2}{2}+1, x\sim\unif{0}{\nicefrac{\pi}{2}}$}\\
		\hline
		\includegraphics[width=0.50\textwidth]{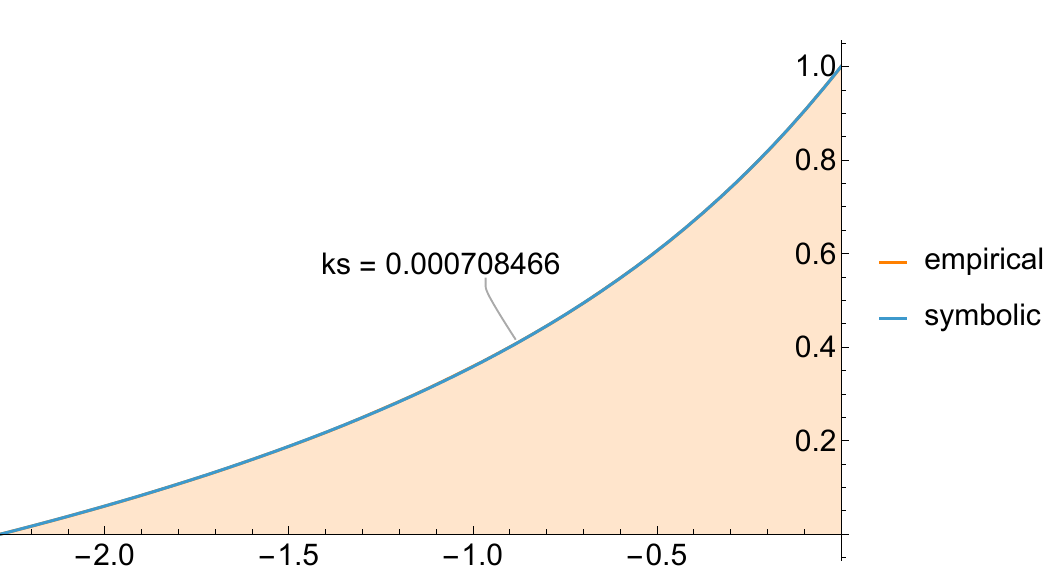} &
		\includegraphics[width=0.5\textwidth]{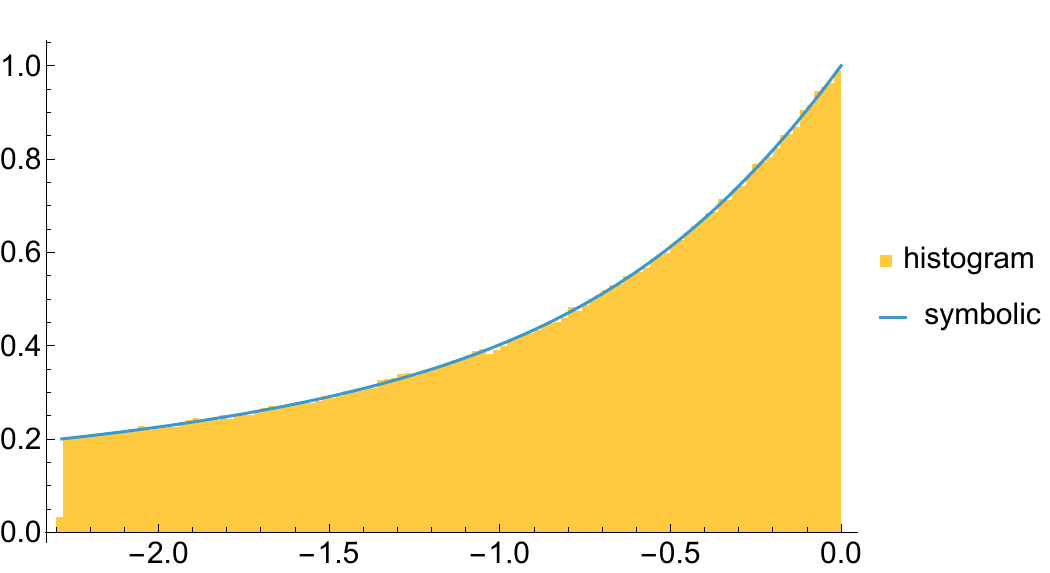} \\
		
		\hline
		\multicolumn{2}{c}{\rule[-0.75em]{0pt}{2em} log: $\frac{1}{5} (x-1)^5-\frac{1}{4} (x-1)^4+\frac{1}{3} (x-1)^3-\frac{1}{2} (x-1)^2+x-1, x\sim\unif{0}{1}$}\\
		\hline 
		
		\multicolumn{2}{c}{\textbf{\small Taylor approximations with triangular inputs}}\\
		\hline\\
		CDF & PDF \\
		\includegraphics[width=0.50\textwidth]{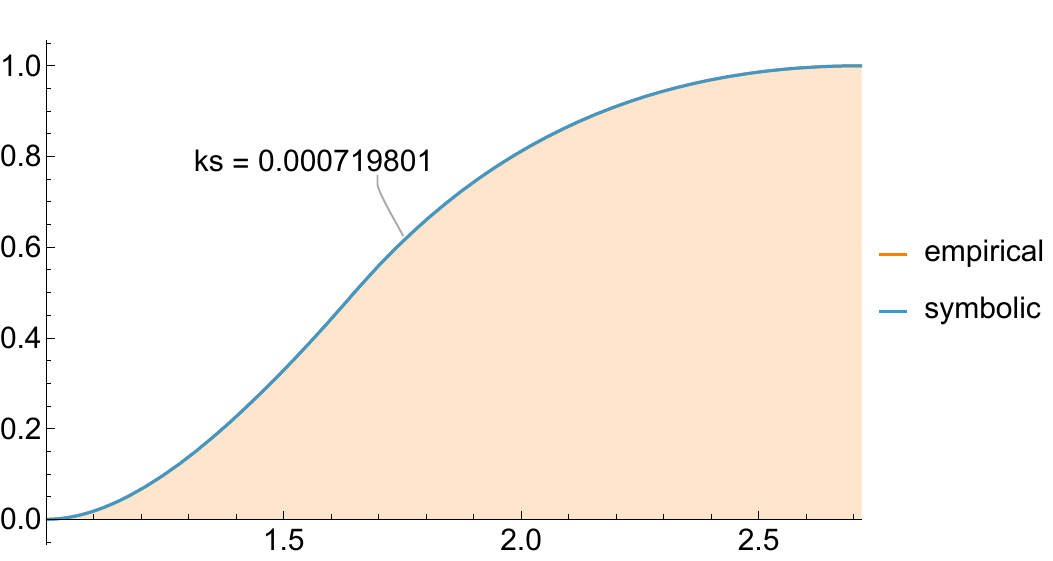} &
		\includegraphics[width=0.5\textwidth]{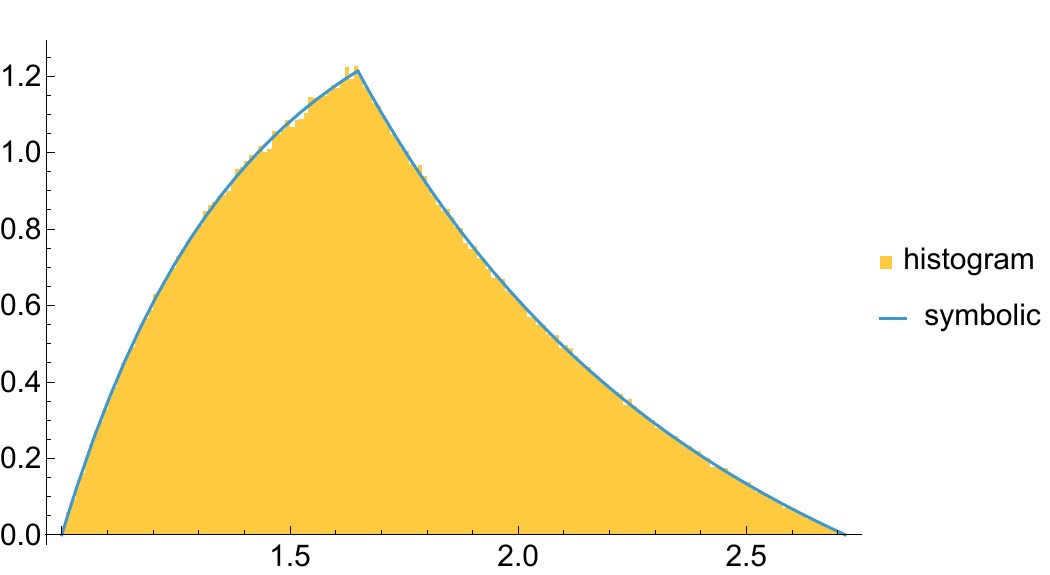} \\
		\hline
		\multicolumn{2}{c}{\rule[-0.75em]{0pt}{2em} exp: $1 + x + \frac{x^2}{2} + \frac{x^3}{6} + \frac{x^4}{24} +\frac{x^5}{120}, x\sim\tri{0,1}$}\\
		\hline
		
		\includegraphics[width=0.50\textwidth]{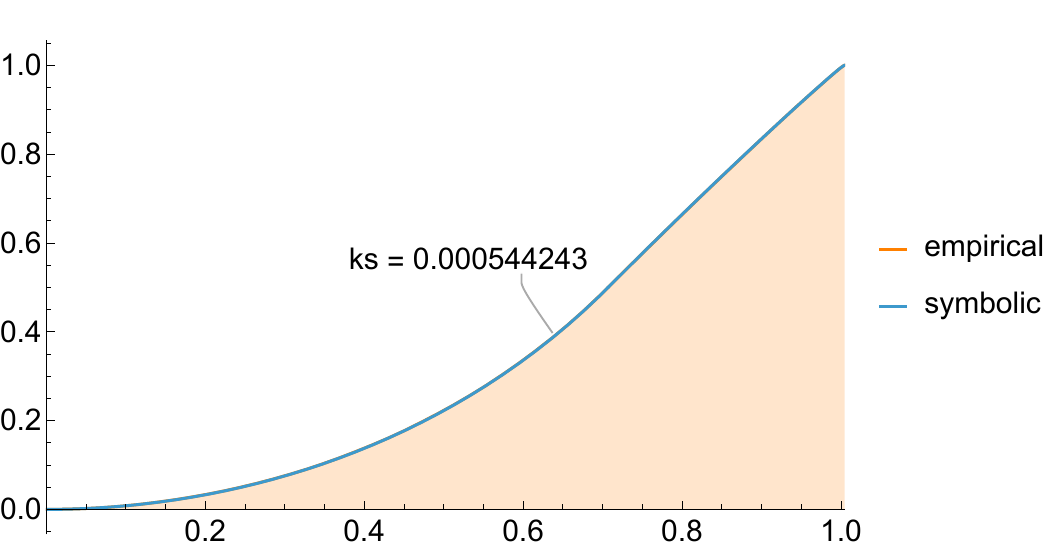} &
		\includegraphics[width=0.5\textwidth]{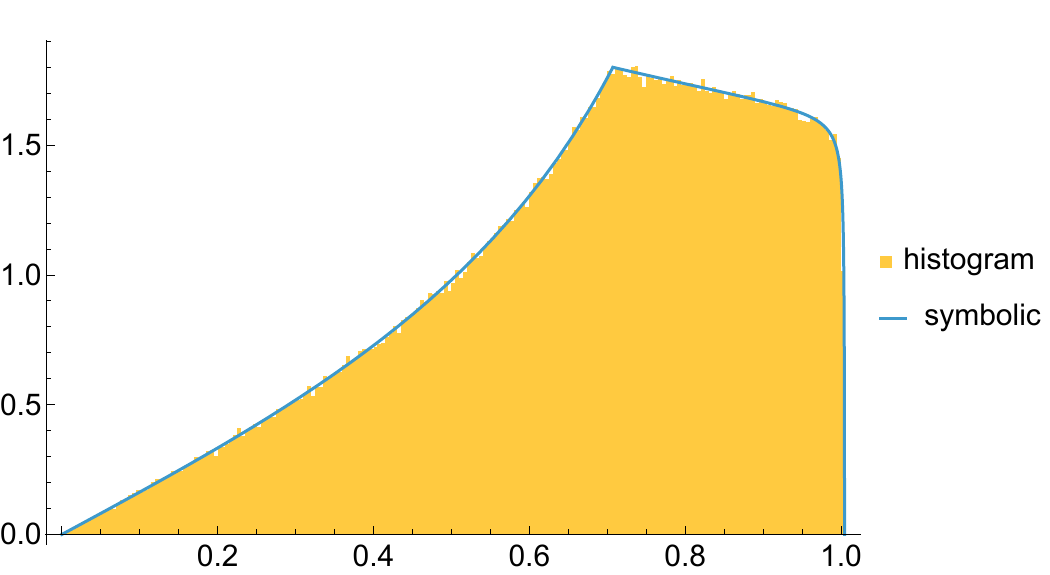} \\
		\hline
		\multicolumn{2}{c}{\rule[-0.75em]{0pt}{2em} sin: $\frac{x^5}{120}-\frac{x^3}{6}+x, x\sim\tri{0}{\nicefrac{\pi}{2}}$}\\
		\hline
		
		\includegraphics[width=0.50\textwidth]{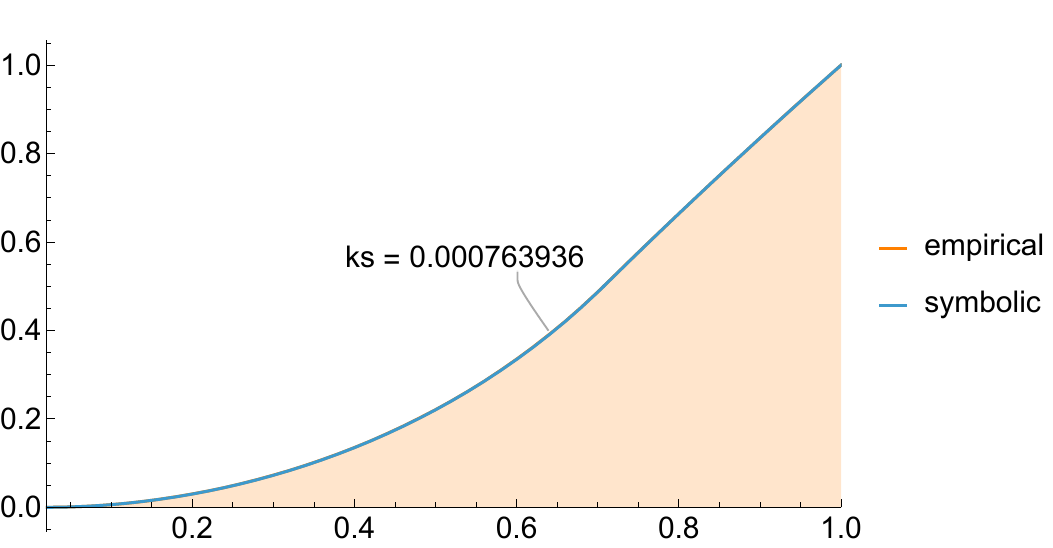} &
		\includegraphics[width=0.5\textwidth]{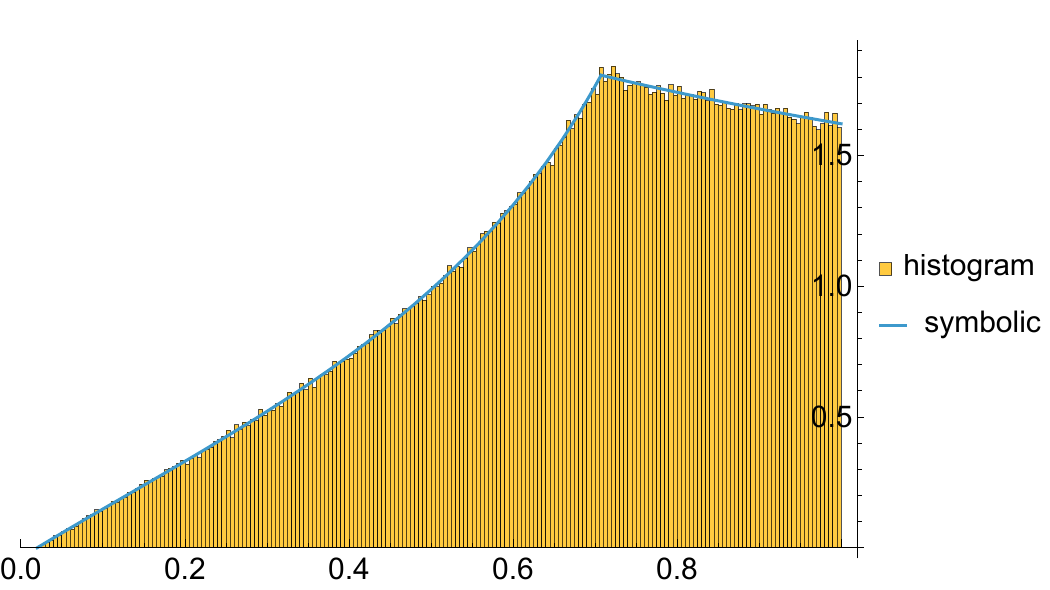} \\
		\hline
		\multicolumn{2}{c}{\rule[-0.75em]{0pt}{2em} cos: $\frac{x^4}{24}-\frac{x^2}{2}+1, x\sim\tri{0}{\nicefrac{\pi}{2}}$}\\
		\hline

		\includegraphics[width=0.50\textwidth]{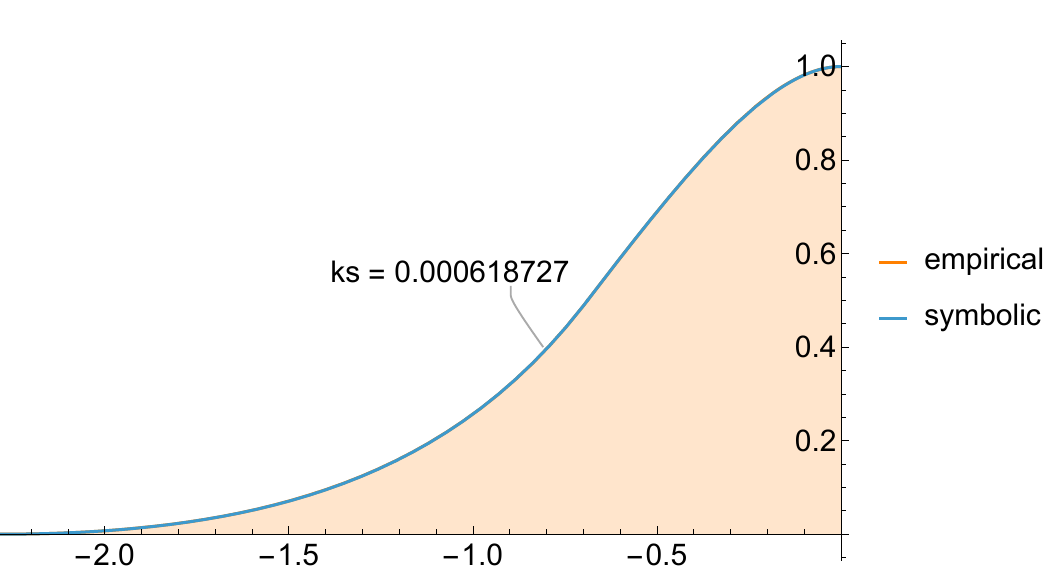} &
		\includegraphics[width=0.5\textwidth]{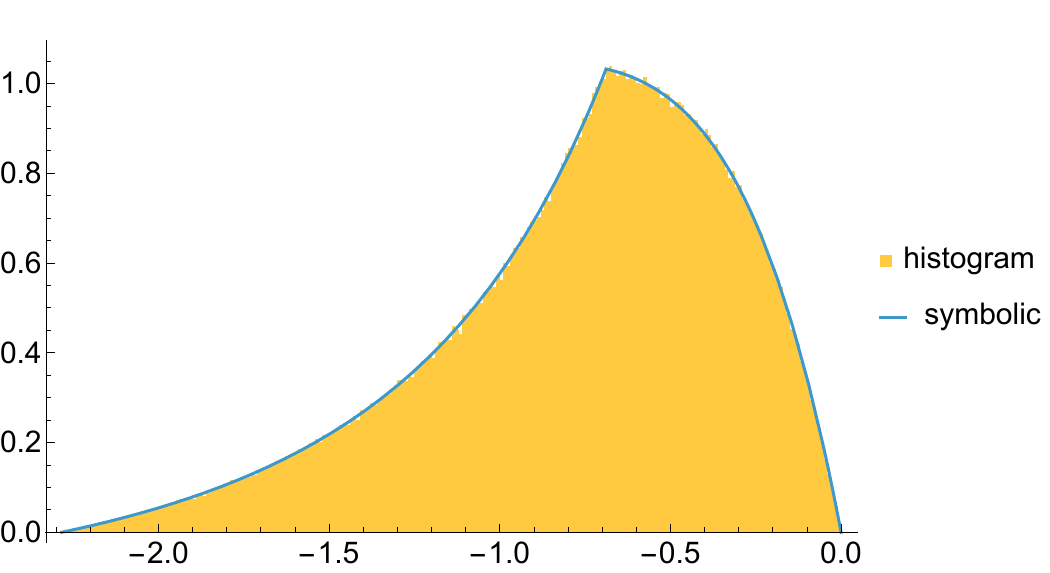} \\
		
		\hline
		\multicolumn{2}{c}{\rule[-0.75em]{0pt}{2em} log: $\frac{1}{5} (x-1)^5-\frac{1}{4} (x-1)^4+\frac{1}{3} (x-1)^3-\frac{1}{2} (x-1)^2+x-1, x\sim\tri{0}{1}$}\\
		\hline \pagebreak
		
		\hline
		\multicolumn{2}{c}{\textbf{\small FPBench benchmarks with uniform inputs}}\\
		\hline\\
		CDF & PDF \\
		
		\includegraphics[width=0.50\textwidth]{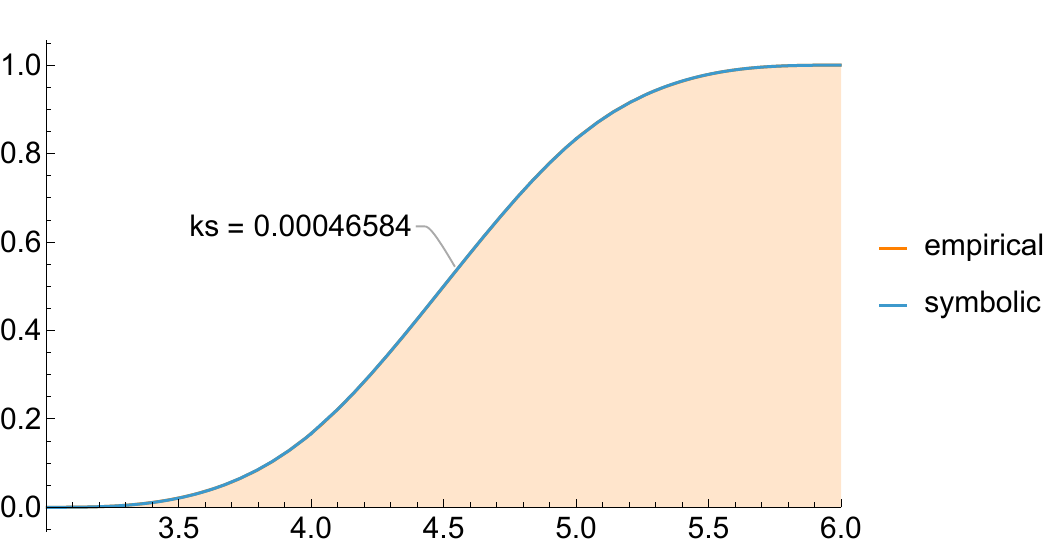} &
		\includegraphics[width=0.5\textwidth]{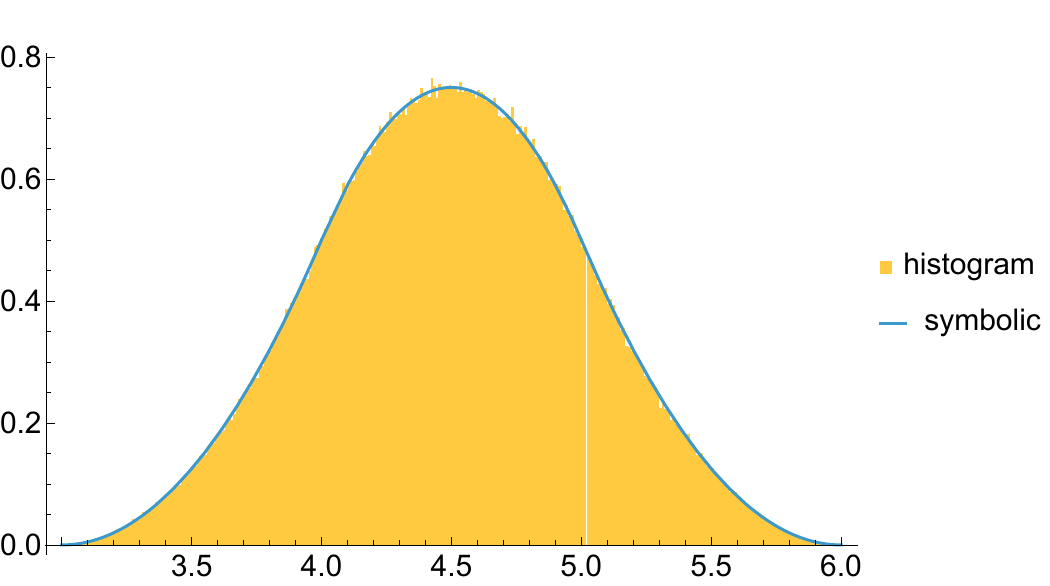} \\
		
		\hline
		\multicolumn{2}{c}{sum: $x0 + x1 + x2,  x0\sim \unif{1}{2}, x1\sim \unif{1}{2}, x2\sim \unif{1}{2}$}\\
		\hline
		\includegraphics[width=0.5\textwidth]{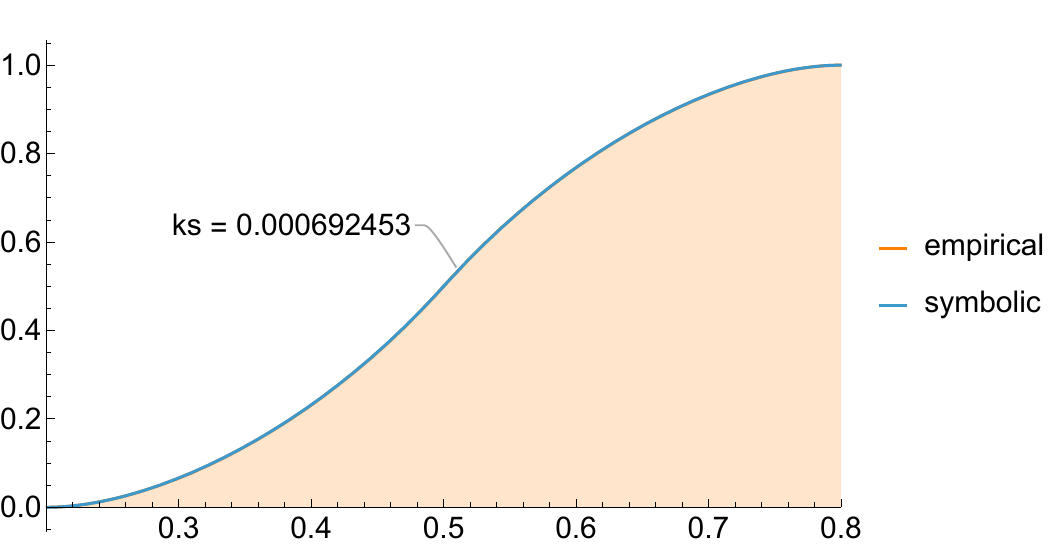} &
		\includegraphics[width=0.5\textwidth]{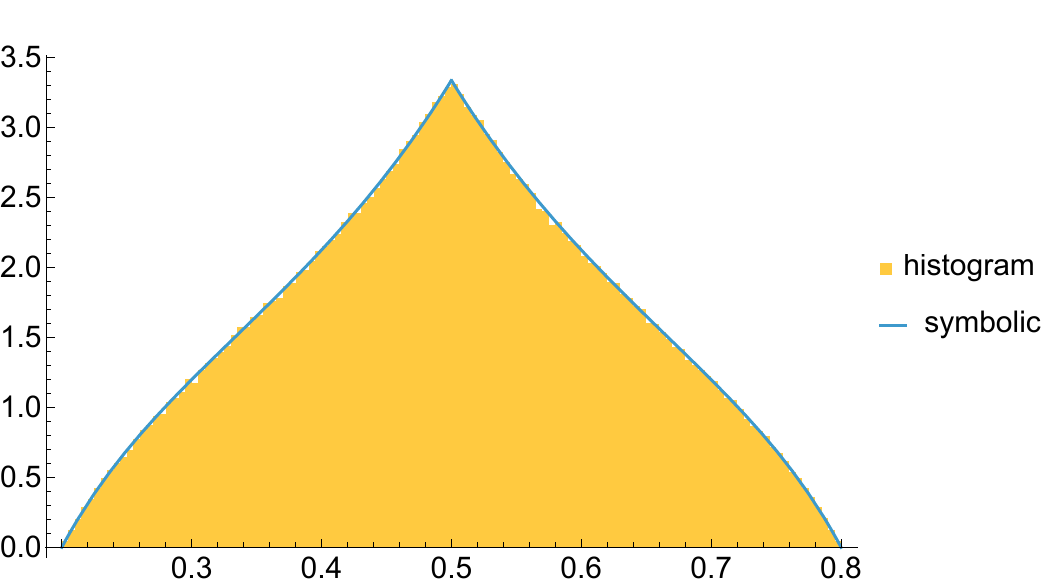} \\
		\hline
		\multicolumn{2}{c}{x\_by\_xy: $\frac{x}{x + y} \text{ with } x\sim\unif{1}{4}, y\sim\unif{1}{4}$} \\
		\hline 
		\includegraphics[width=0.5\textwidth]{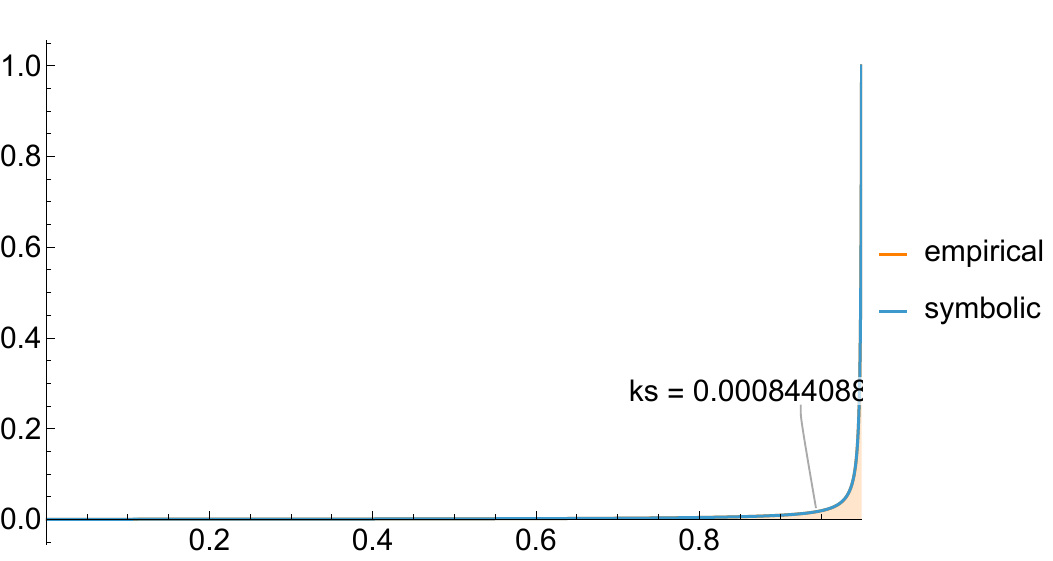} &
		\includegraphics[width=0.5\textwidth]{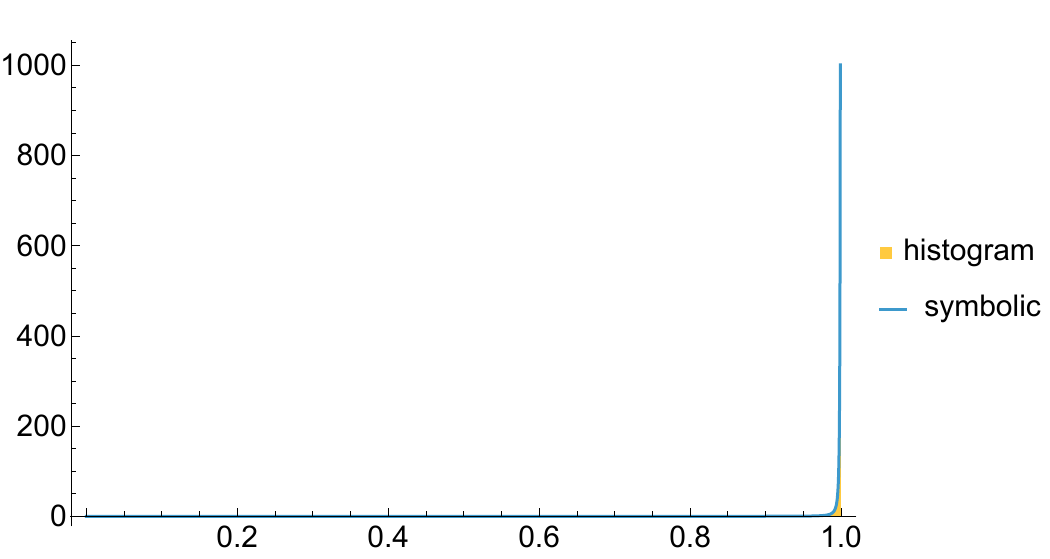} \\
		\hline
		\multicolumn{2}{c}{nonlin1: $\frac{z}{z+1}, z\sim\unif{0}{999}$} \\
		\hline
		\includegraphics[width=0.5\textwidth]{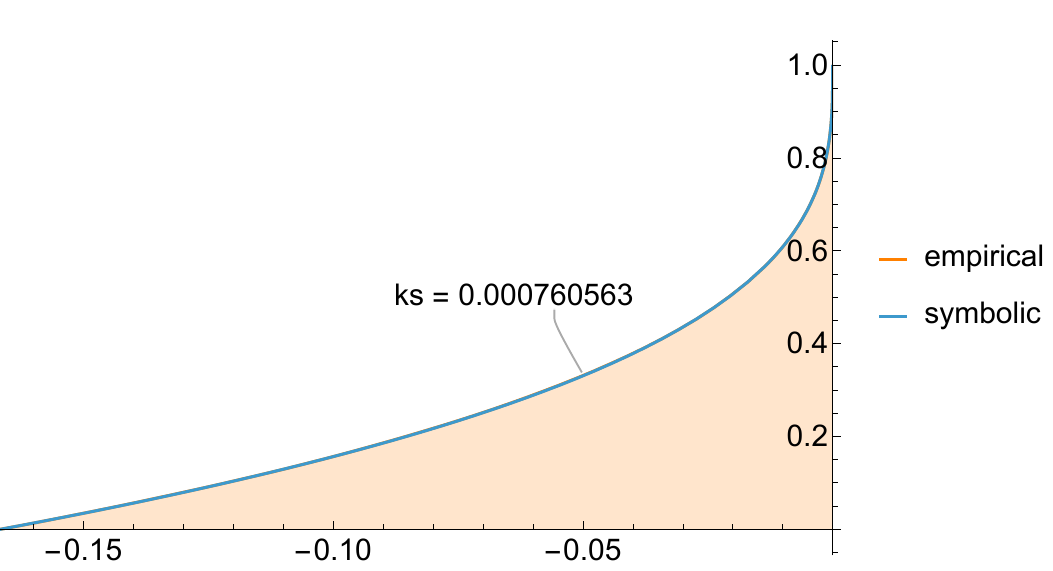} &
		\includegraphics[width=0.5\textwidth]{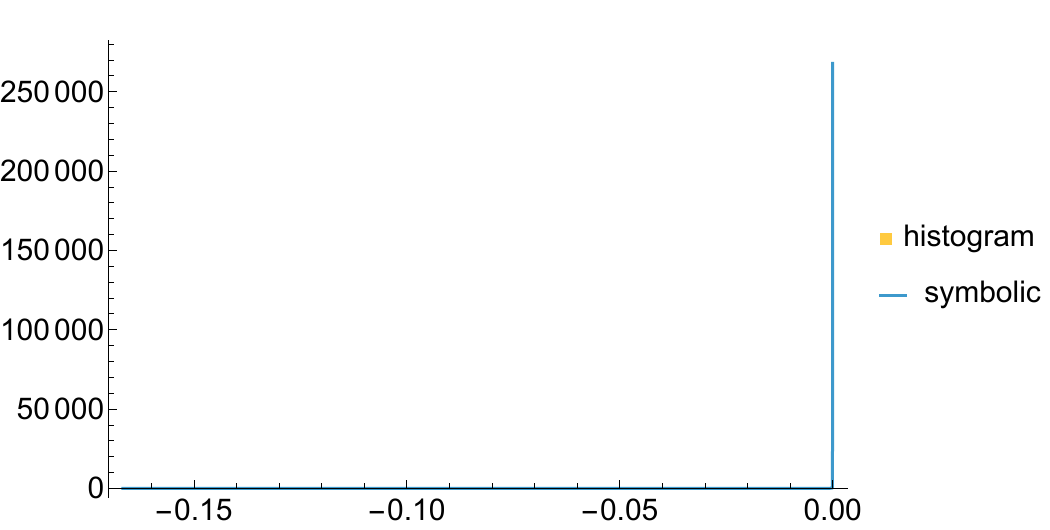} \\
		\hline
		\multicolumn{2}{c}{bspline3: $-\frac{u^3}{6}, u\sim\unif{0}{1}$} \\
		\hline
		\includegraphics[width=0.5\textwidth]{pictures/fpuniform/cdftestcav10} &
		\includegraphics[width=0.5\textwidth]{pictures/fpuniform/pdftestcav10} \\
		\hline
		\multicolumn{2}{c}{cav10: $\text{If}\left[x^2-x\geq 0,\frac{x}{10},x^2+2\right], x\sim\unif{0}{10}$} \\
		\hline
		
		\includegraphics[width=0.5\textwidth]{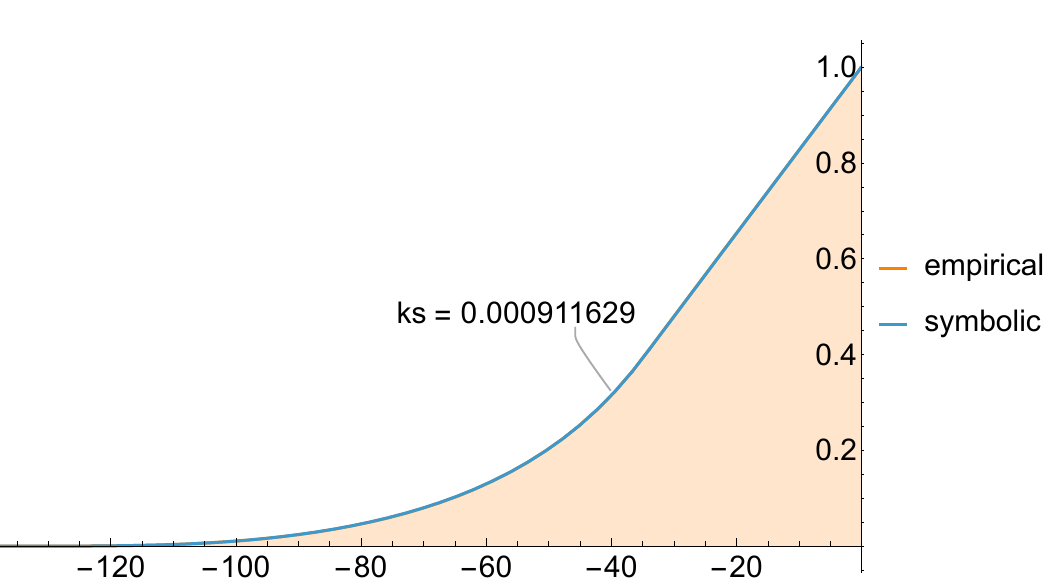} &
		\includegraphics[width=0.5\textwidth]{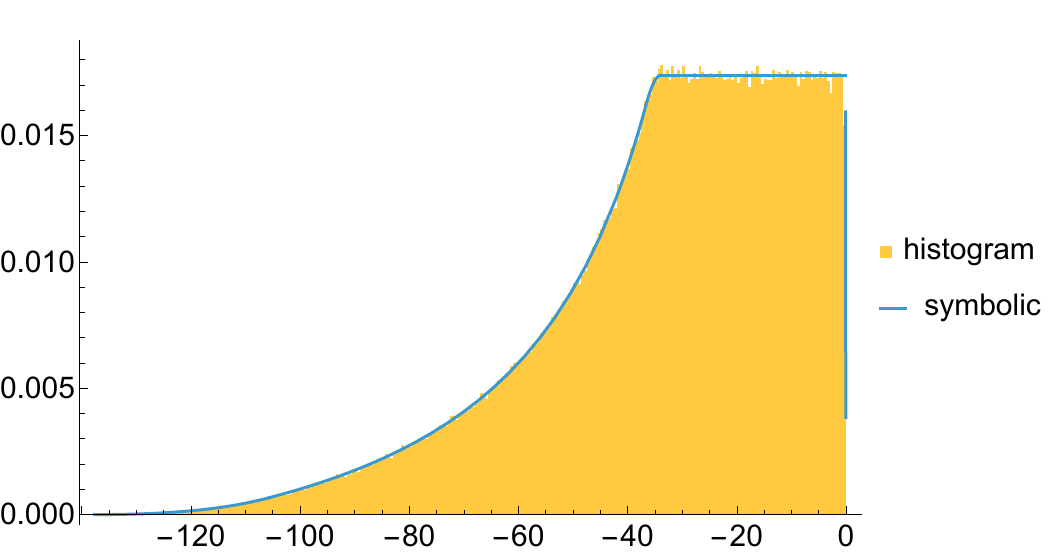} \\
		\hline
		\multicolumn{2}{c}{doppler1: $\frac{\left(-\frac{3 T}{5}-\frac{1657}{5}\right) v}{\left(\frac{3 T}{5}+u+\frac{1657}{5}\right)^2}, u\sim\unif{-100}{100}, v\sim\unif{20}{20000}, T\sim\unif{-30}{50}$} \\
		\hline
		
		\includegraphics[width=0.5\textwidth]{pictures/fpuniform/cdftestdoppler2} &
		\includegraphics[width=0.5\textwidth]{pictures/fpuniform/pdftestdoppler2} \\
		\hline
		\multicolumn{2}{c}{doppler2: $\frac{\left(-\frac{3 T}{5}-\frac{1657}{5}\right) v}{\left(\frac{3 T}{5}+u+\frac{1657}{5}\right)^2}, u\sim\unif{-125}{125}, v\sim\unif{15}{25000}, T\sim\unif{-40}{60}$} \\
		\hline
		
		\includegraphics[width=0.5\textwidth]{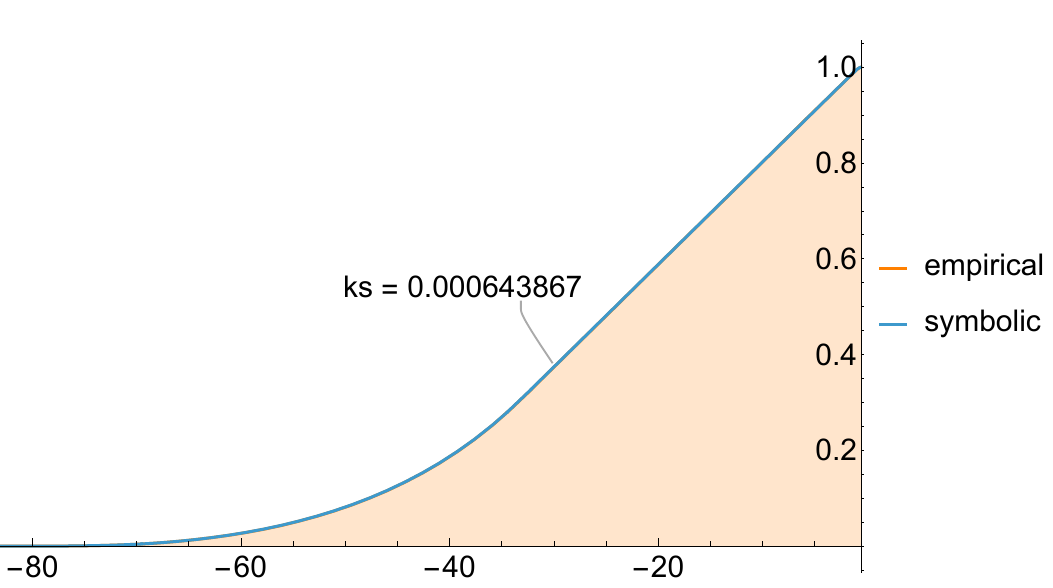} &
		\includegraphics[width=0.5\textwidth]{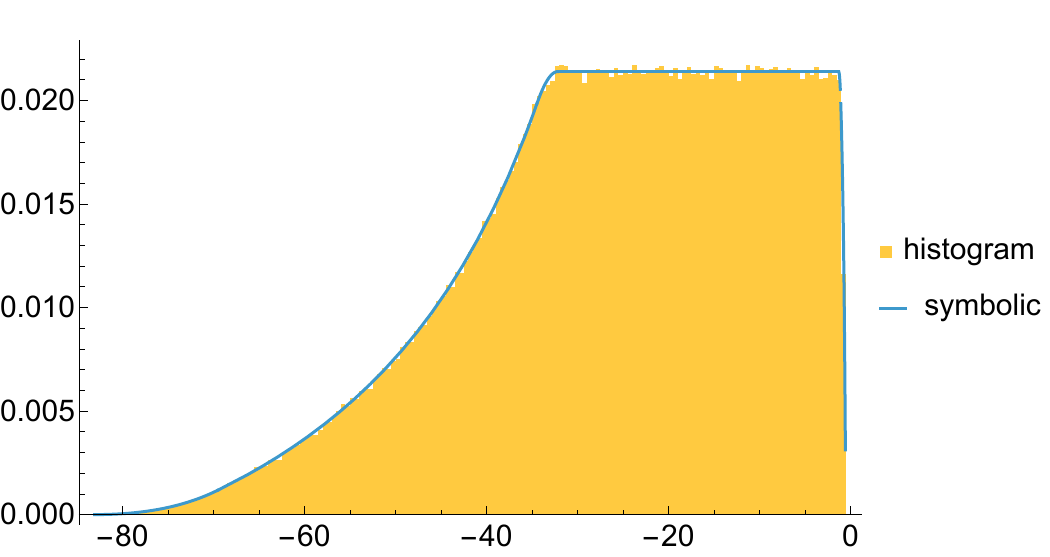} \\
		\hline
		\multicolumn{2}{c}{doppler3: $\frac{\left(-\frac{3 T}{5}-\frac{1657}{5}\right) v}{\left(\frac{3 T}{5}+u+\frac{1657}{5}\right)^2}, u\sim\unif{-30}{120}, v\sim\unif{320}{20300}, T\sim\unif{-50}{30}$} \\
		\hline
		
		\includegraphics[width=0.5\textwidth]{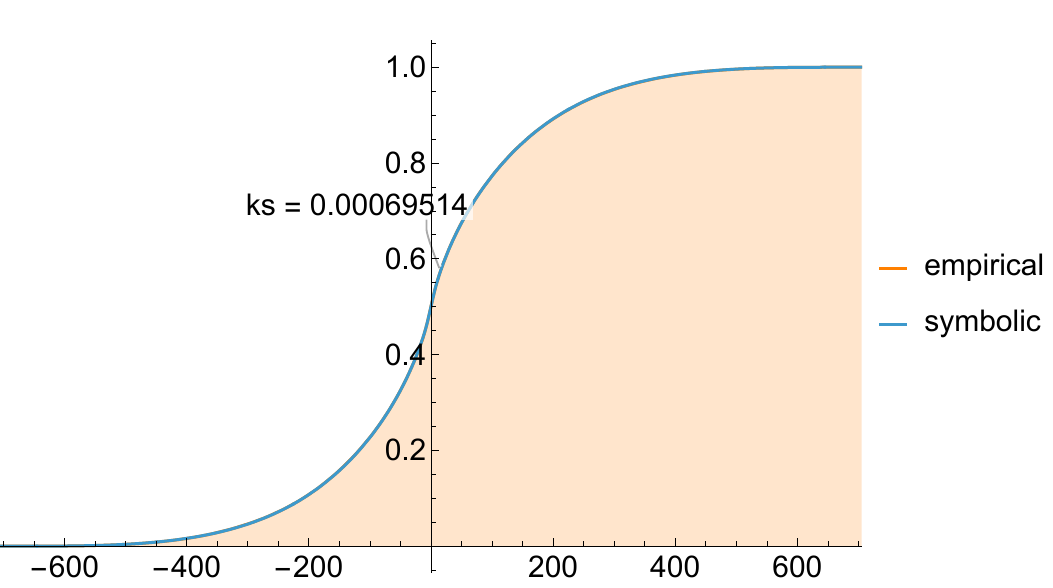} &
		\includegraphics[width=0.5\textwidth]{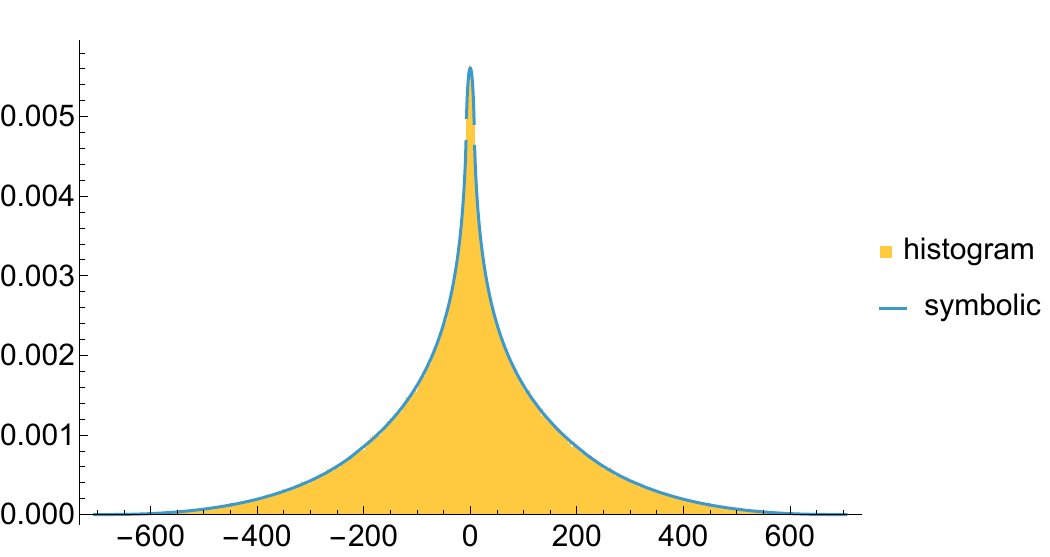} \\
		\hline
		\multicolumn{2}{c}{rigidBody1: $-x_1 - x_1 x_2 - x_3 - 2 x_2 x_3, x_1\sim\unif{-15}{15},x_2\sim\unif{-15}{15},x_3\sim\unif{-15}{15} $} \\
		\hline
		\\
		\multicolumn{2}{c}{\textbf{\small FPBench benchmarks with triangular inputs}}\\
		\hline\\
		CDF & PDF \\
		
		\includegraphics[width=0.50\textwidth]{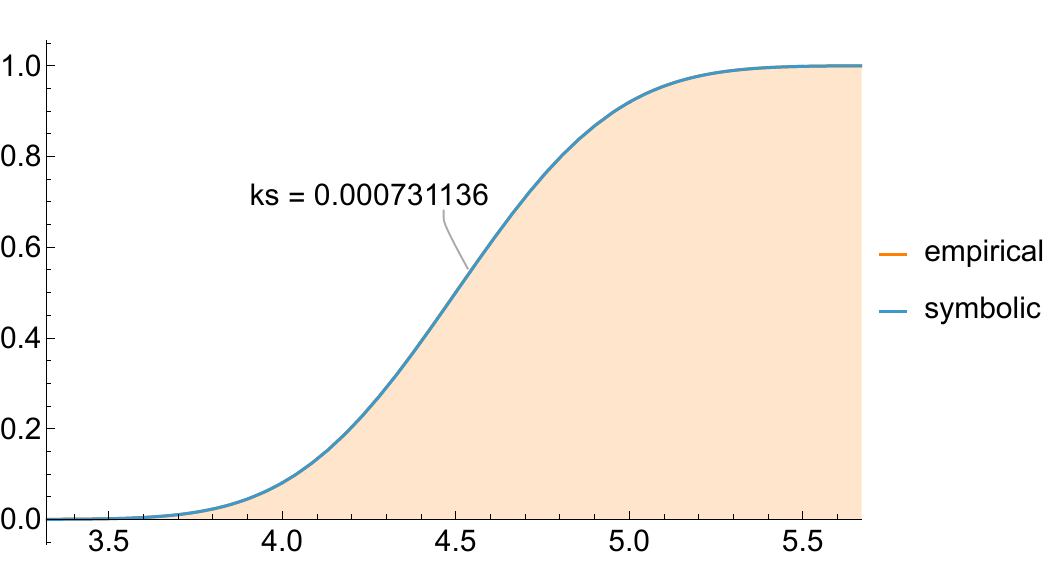} &
		\includegraphics[width=0.5\textwidth]{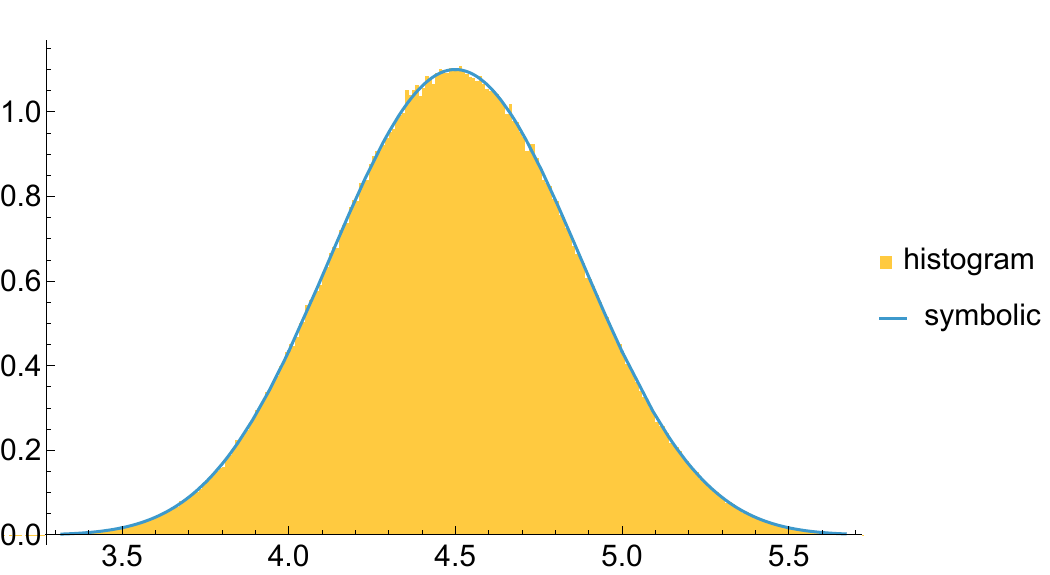} \\
		
		\hline
		\multicolumn{2}{c}{sum: $x0 + x1 + x2, x0\sim \tri{1}{2}, x1\sim\tri{1}{2}, x2\sim\tri{1}{2}$}\\
		\hline
		\includegraphics[width=0.5\textwidth]{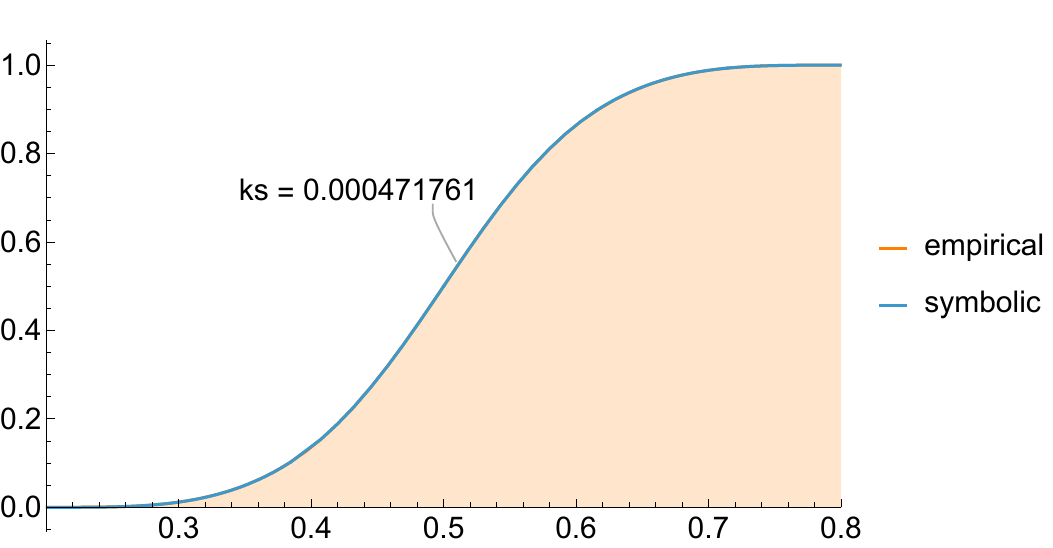} &
		\includegraphics[width=0.5\textwidth]{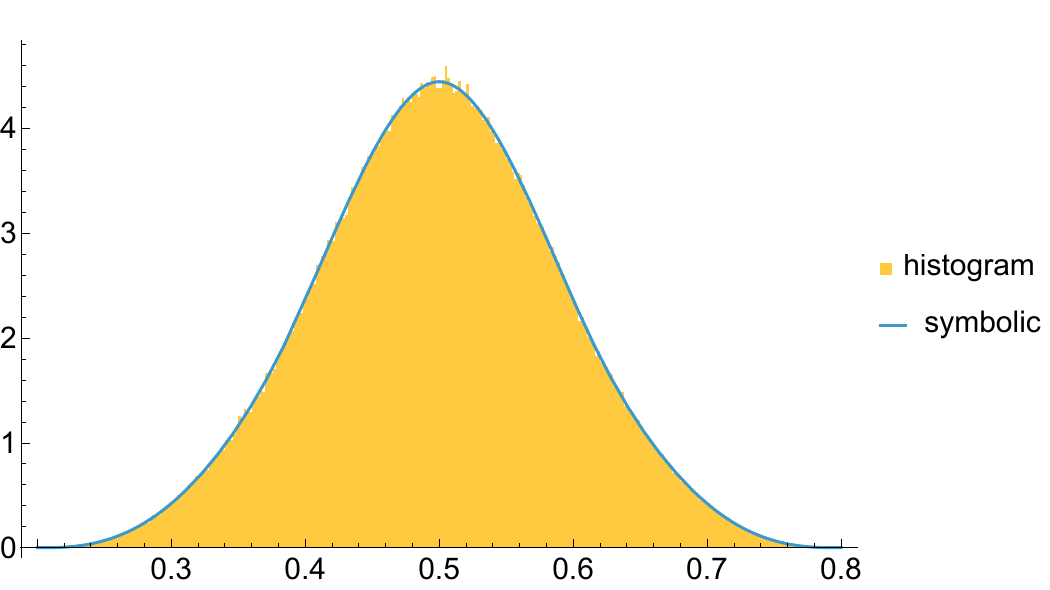} \\
		\hline
		\multicolumn{2}{c}{x\_by\_xy: $\frac{x}{x + y}, x\tri{1}{4}, y\tri{1}{4}$} \\
		\hline 
 		\includegraphics[width=0.5\textwidth]{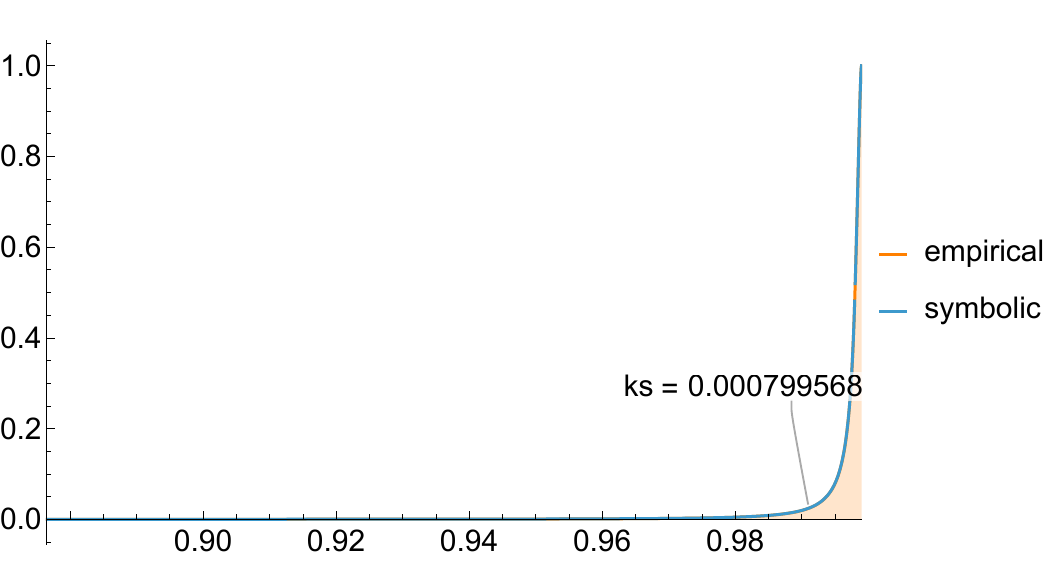} &
		\includegraphics[width=0.5\textwidth]{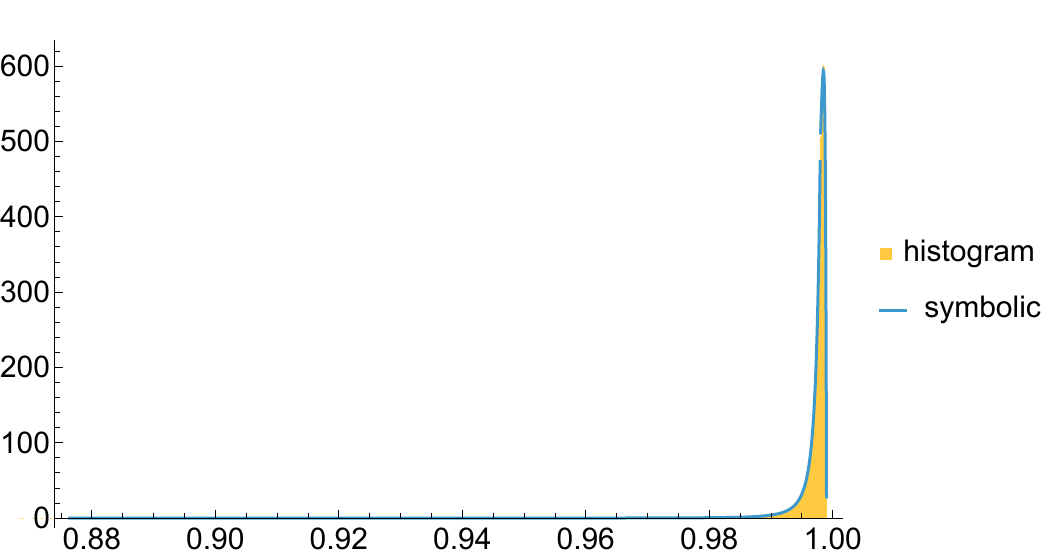} \\
		\hline
		\multicolumn{2}{c}{nonlin1: $\frac{z}{z+1}, x\sim\tri{0}{999}$} \\
		\hline
		\includegraphics[width=0.5\textwidth]{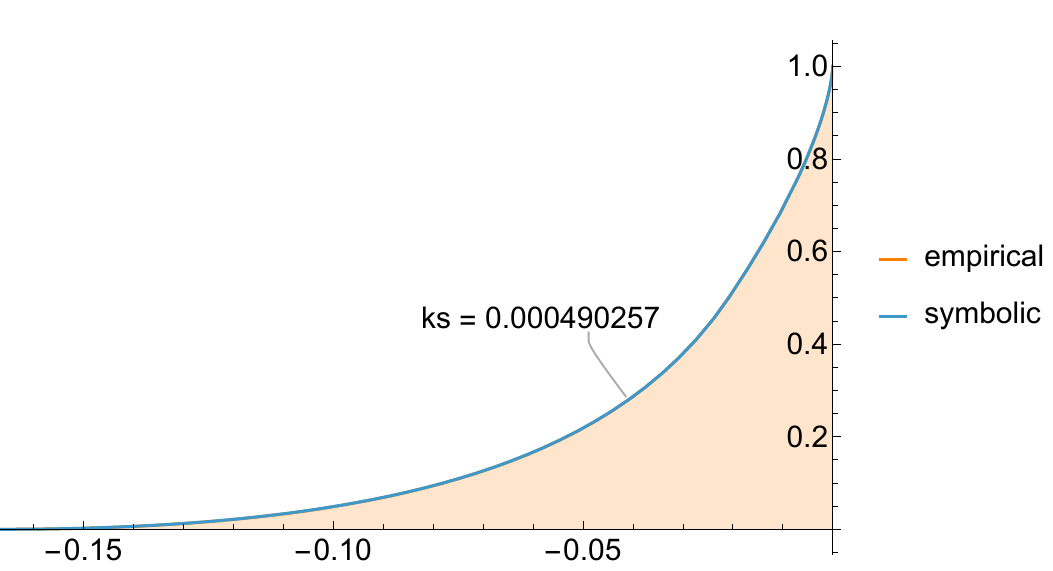} &
		\includegraphics[width=0.5\textwidth]{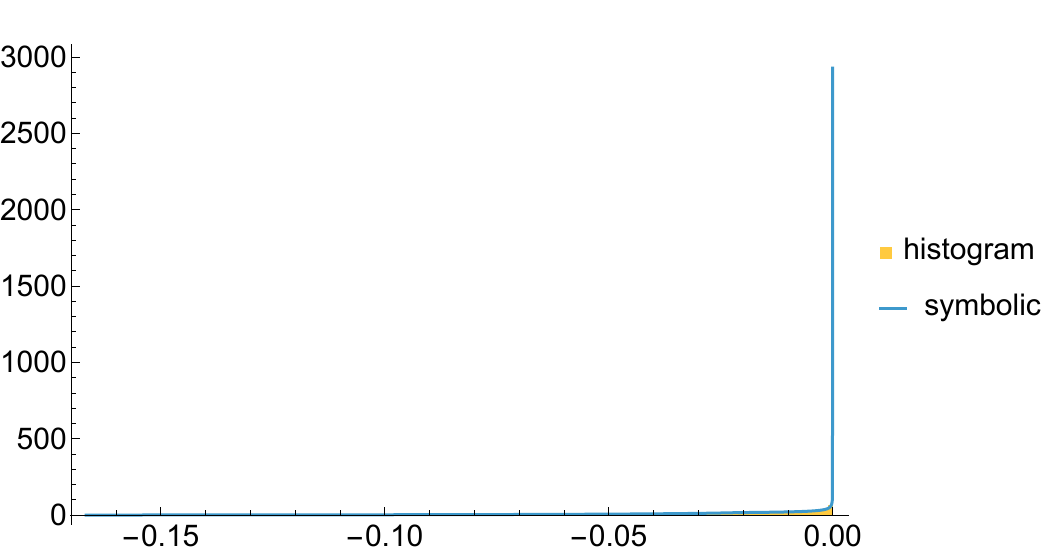} \\
		\hline
		\multicolumn{2}{c}{bspline3: $-\frac{u^3}{6}, u\sim\tri{0}{1}$} \\
		\hline
		\includegraphics[width=0.5\textwidth]{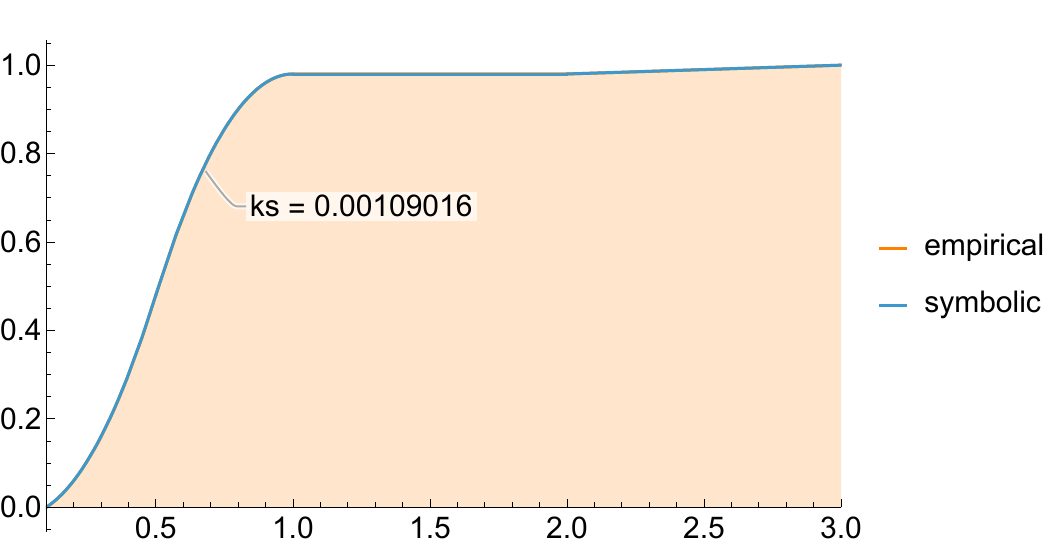} &
		\includegraphics[width=0.5\textwidth]{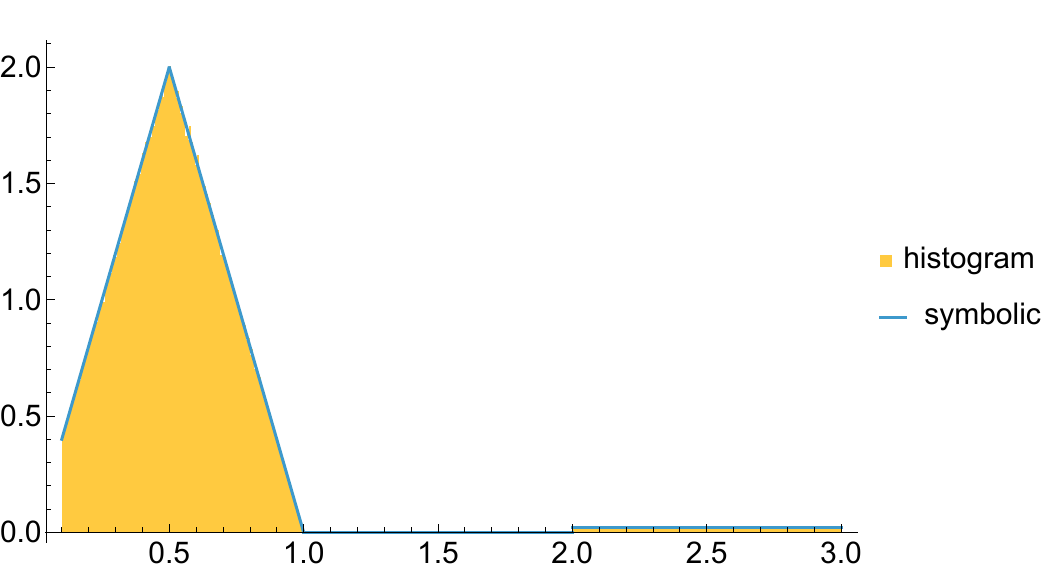} \\
		\hline
		\\
		\multicolumn{2}{c}{\textbf{\small FPBench benchmarks with beta inputs}}\\
		\hline\\
		CDF & PDF \\
		
		\includegraphics[width=0.50\textwidth]{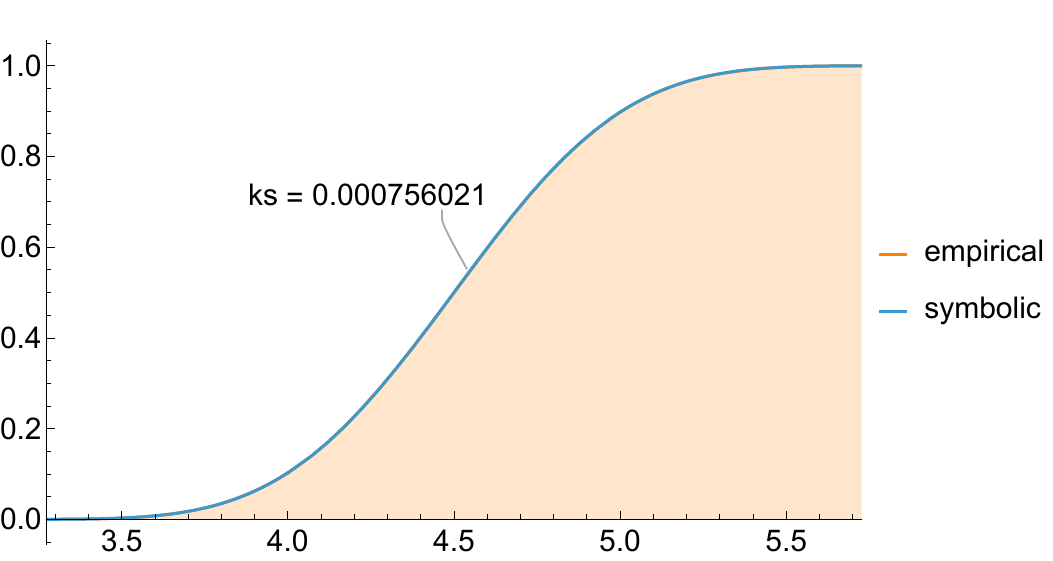} &
		\includegraphics[width=0.5\textwidth]{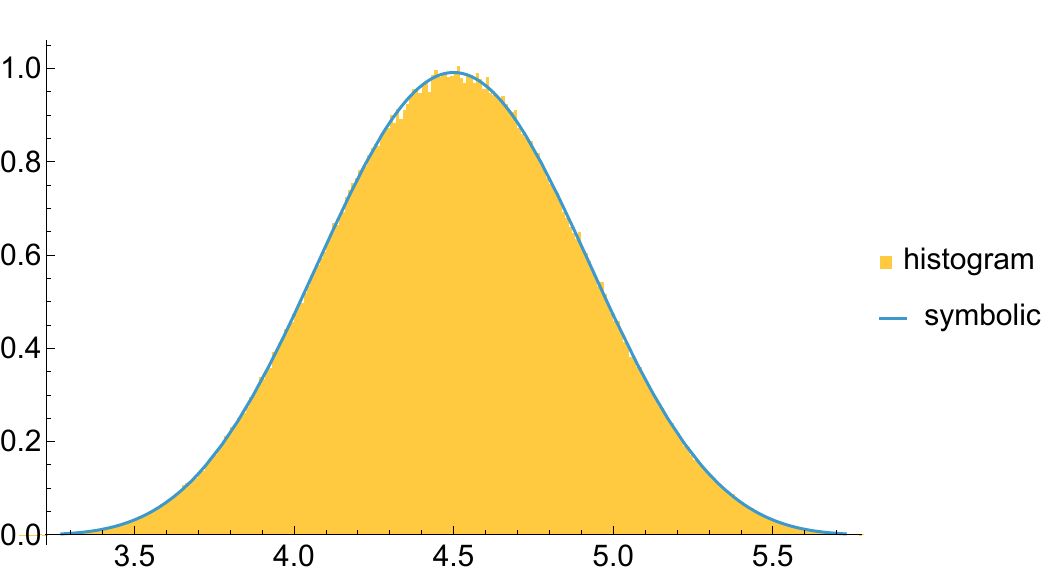} \\
		
		\hline
		\multicolumn{2}{c}{sum: $x0 + x1 + x2, x0\sim \tri{1}{2}, x1\sim\tri{1}{2}, x2\sim\tri{1}{2}$}\\
		\hline
		\includegraphics[width=0.5\textwidth]{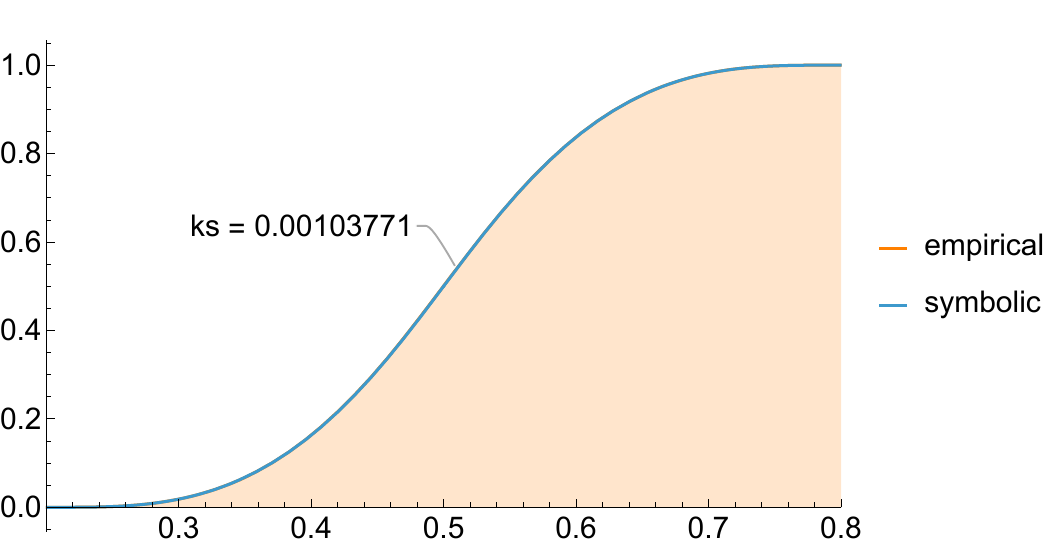} &
		\includegraphics[width=0.5\textwidth]{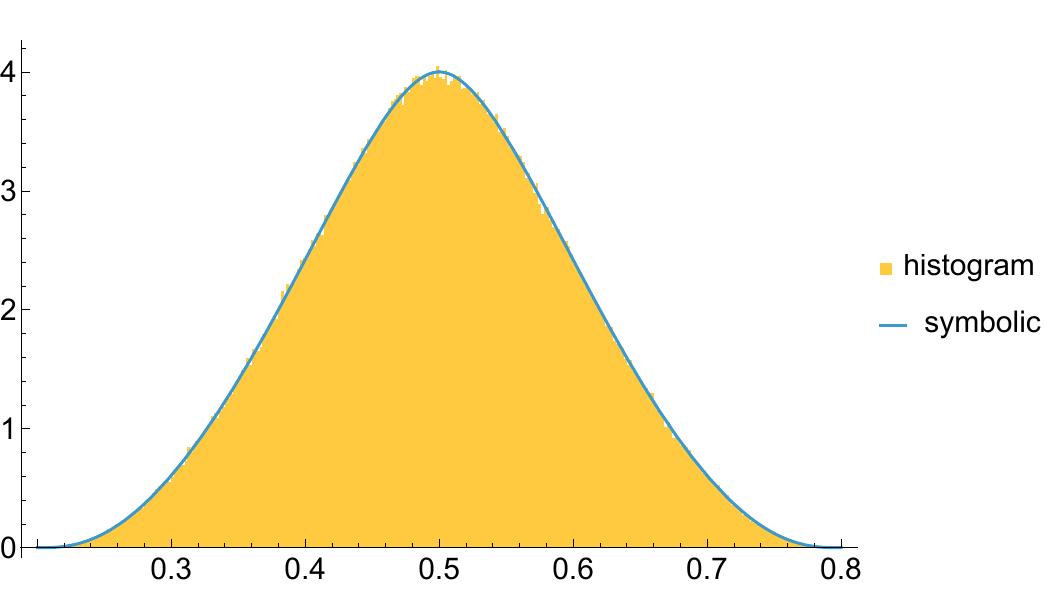} \\
		\hline
		\multicolumn{2}{c}{x\_by\_xy: $\frac{x}{x + y}, x\tri{1}{4}, y\tri{1}{4}$} \\
		\hline 
		\includegraphics[width=0.5\textwidth]{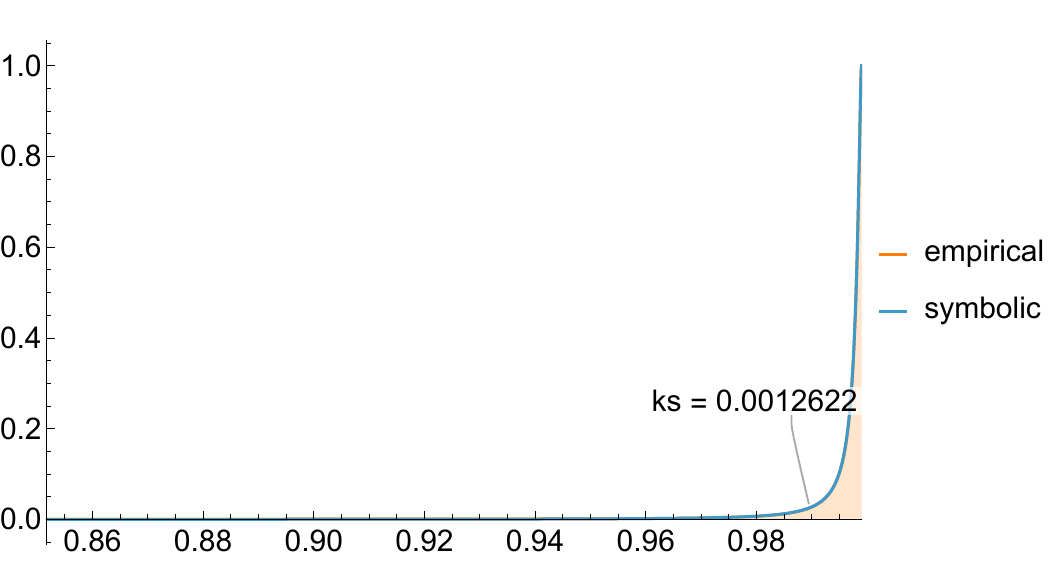} &
		\includegraphics[width=0.5\textwidth]{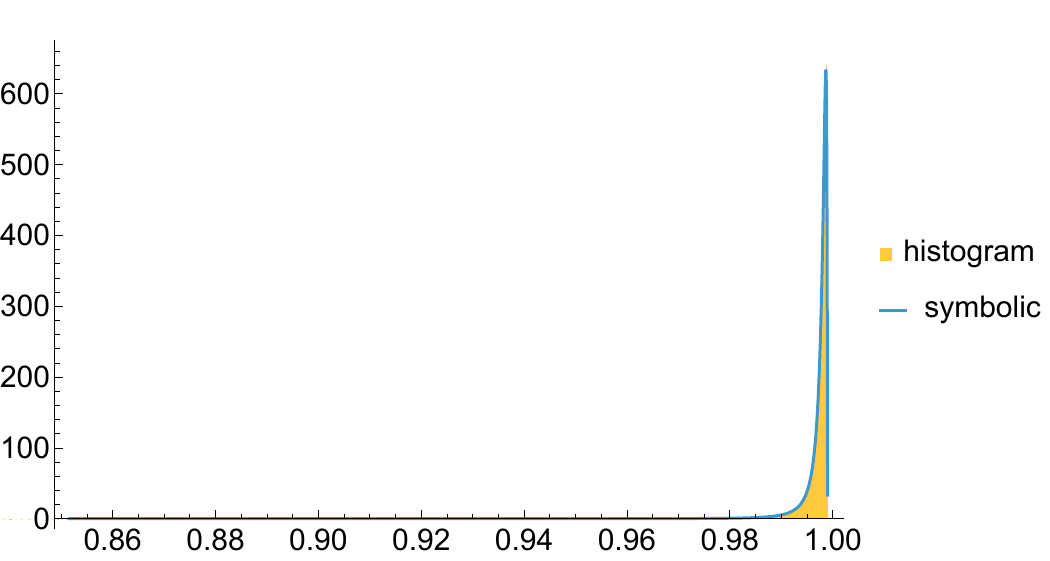} \\
		\hline
		\multicolumn{2}{c}{nonlin1: $\frac{z}{z+1}, x\sim\tri{0}{999}$} \\
		\hline
		\includegraphics[width=0.5\textwidth]{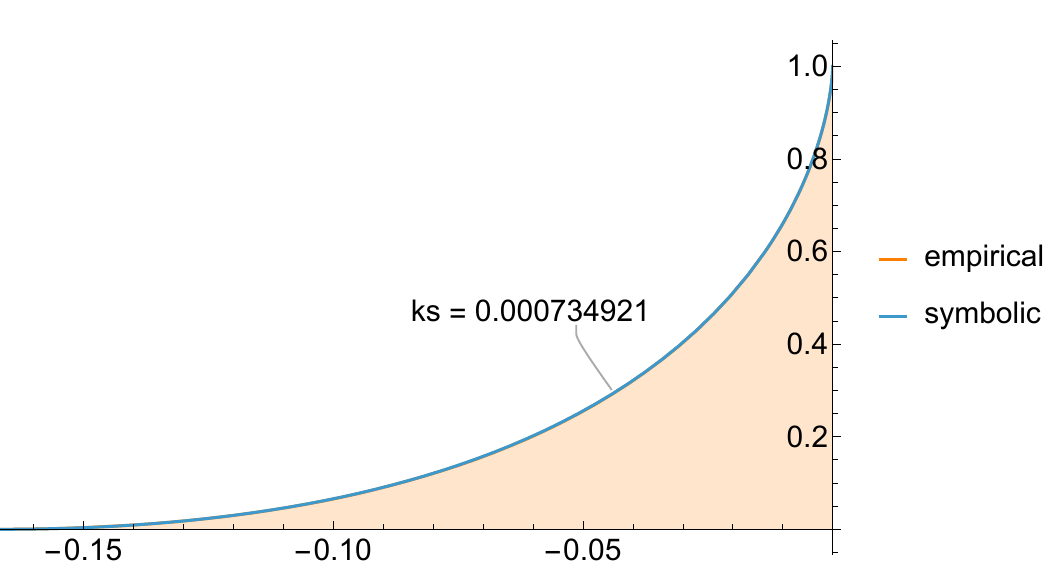} &
		\includegraphics[width=0.5\textwidth]{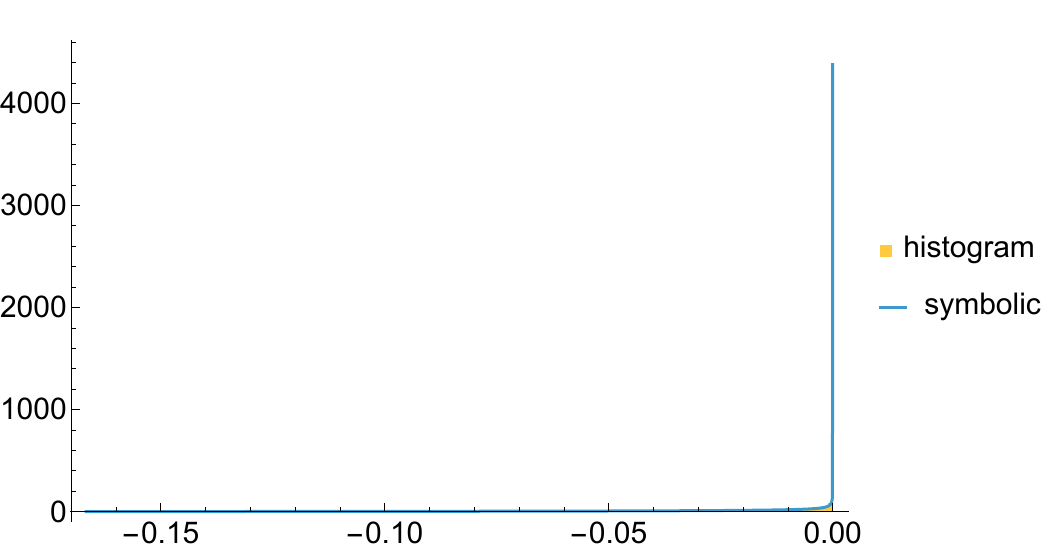} \\
		\hline
		\multicolumn{2}{c}{bspline3: $-\frac{u^3}{6}, u\sim\tri{0}{1}$} \\
		\hline
		\includegraphics[width=0.5\textwidth]{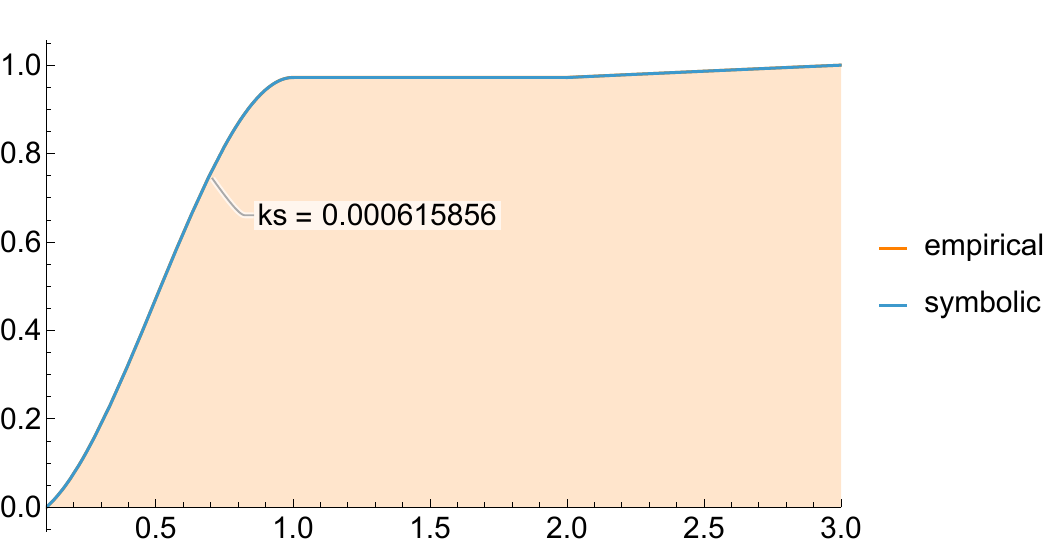} &
		\includegraphics[width=0.5\textwidth]{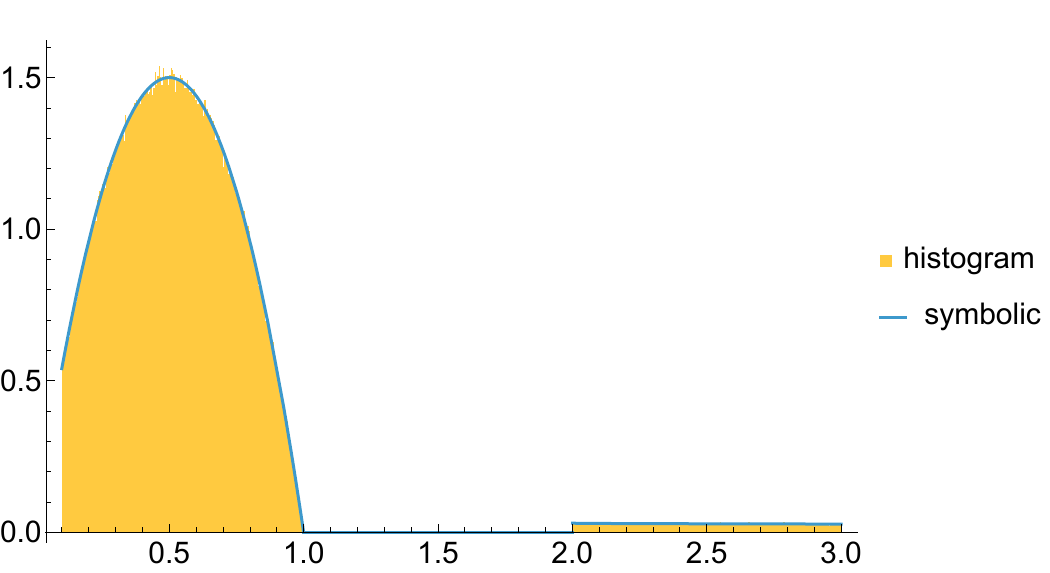} \\
		\hline
		\\
		\multicolumn{2}{c}{\textbf{\small IMU benchmarks with uniform inputs}}\\
		\hline\\
		CDF & PDF \\
		
		\hline
		\multicolumn{2}{c}{cav10: $\text{If}\left[x^2-x\geq 0,\frac{x}{10},x^2+2\right] , x\sim \tri{0}{10}$} \\
		\hline
		\includegraphics[width=0.5\textwidth]{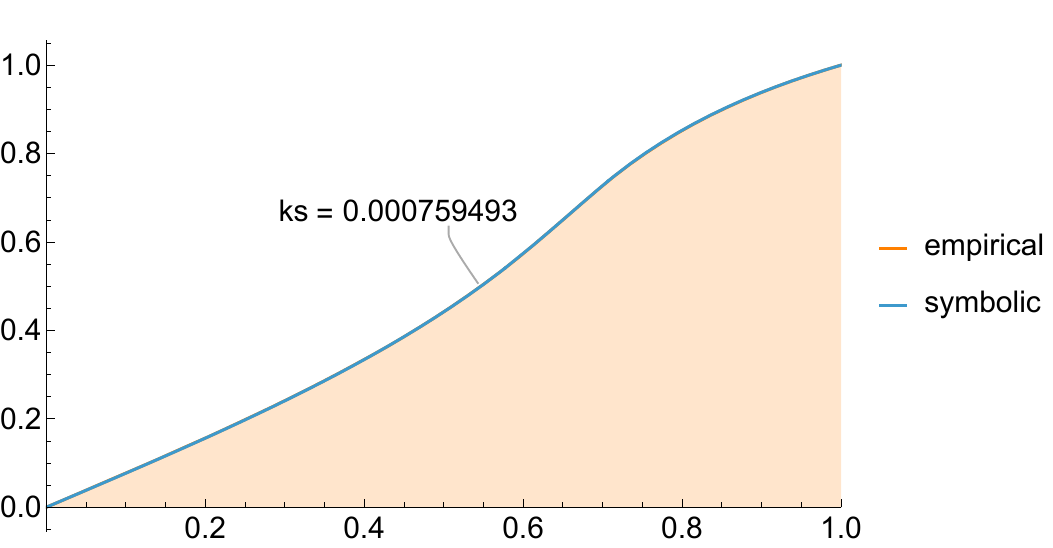} &
		\includegraphics[width=0.5\textwidth]{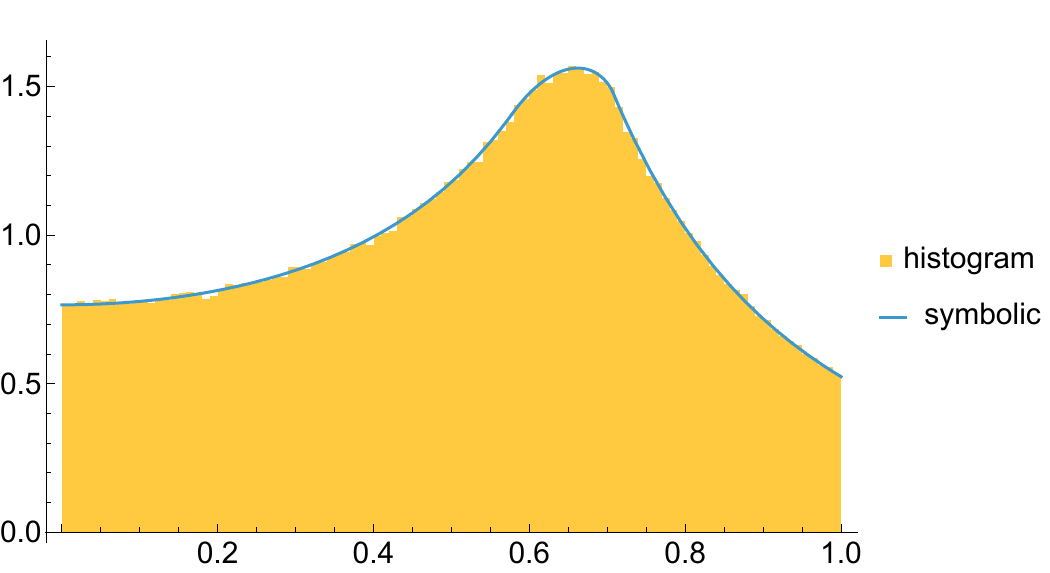} \\
		\hline
		\multicolumn{2}{c}{vector\_norm: $\frac{x}{\sqrt{x^2 + y^2 + z^2}} , x\sim \unif{0}{1}, y\sim \unif{0}{1}, z\sim \unif{0}{1}$} \\
		\hline
		\includegraphics[width=0.5\textwidth]{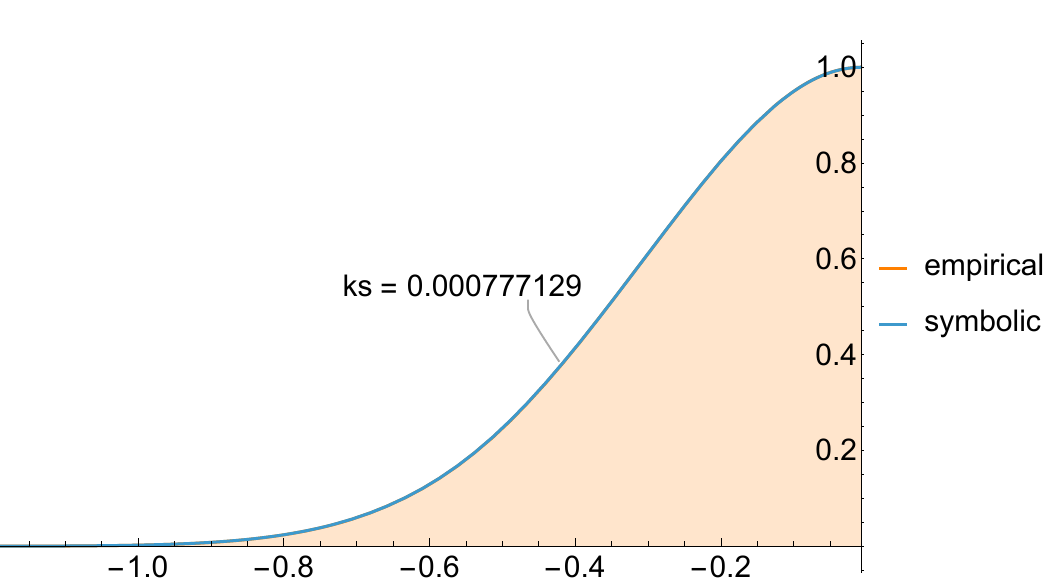} &
		\includegraphics[width=0.5\textwidth]{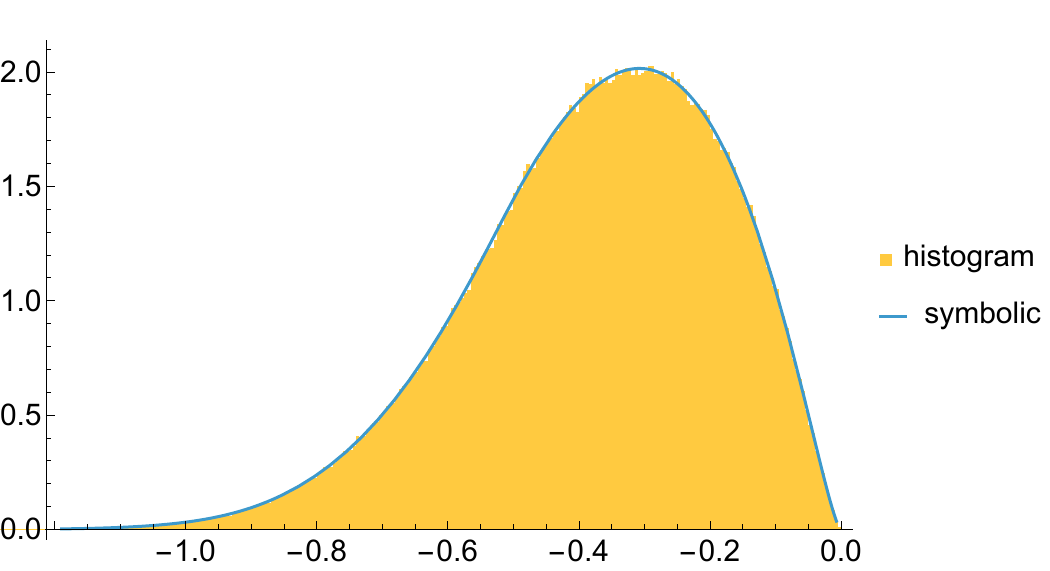} \\
		\hline
		\multicolumn{2}{c}{change\_gyro: $0.5*(-q1 * gx - q2 * gy - q3 * gz)$,}\\
		
		\multicolumn{2}{c}{$q1\sim \unif{0}{1}, q2\sim \unif{0}{1}, q3\sim \unif{0}{1}, gx\sim \unif{0}{1}, gy\sim \unif{0}{1}, gz\sim \unif{0}{1}$} \\
		\hline
	
	\end{longtable}
}

\end{document}